\documentclass[final]{siamltex}

\usepackage{amsmath,amssymb}
\usepackage[norelsize,ruled,vlined]{algorithm2e}
\usepackage{latexsym}
\usepackage{graphicx}
\usepackage{color}
\usepackage{subfigure}
\usepackage{hyperref}
\usepackage{multirow}
\usepackage{url}

\newcommand{\RR}{\mathbb{R}}

\newcommand{\TT}{\mathbb{T}}

\newcommand{\CC}{\mathbb{C}}

\newcommand{\III}{\mathcal{I}}

\newtheorem{lem}[theorem]{Lemma}

\newcommand{\vol}{\operatorname{vol}}

\newcommand{\pu}{\mathrm{PU}}

\newcommand{\argmax}{\operatornamewithlimits{argmax}}

\begin{document}

\title{Phase retrieval with polarization}




\author{
Boris Alexeev\thanks{Department of Mathematics,
Princeton University, Princeton, NJ 08544, USA ({\tt balexeev@math.princeton.edu}).}
\and
Afonso S.\ Bandeira\thanks{Program in Applied and Computational Mathematics (PACM),
Princeton University, Princeton, NJ 08544, USA ({\tt ajsb@math.princeton.edu}).}
\and
Matthew Fickus\thanks{Department of Mathematics and Statistics, Air Force Institute of Technology, Wright-Patterson Air Force Base, OH 45433, USA ({\tt matthew.fickus@gmail.com}).}
\and
Dustin G.\ Mixon\thanks{Department of Mathematics and Statistics, Air Force Institute of Technology, Wright-Patterson Air Force Base, OH 45433, USA ({\tt dustin.mixon@afit.edu}).}
}

\maketitle

{\small
\begin{abstract}
In many areas of imaging science, it is difficult to measure the phase of linear measurements.
As such, one often wishes to reconstruct a signal from intensity measurements, that is, perform \textit{phase retrieval}.
In this paper, we provide a novel measurement design which is inspired by interferometry and exploits certain properties of expander graphs.
We also give an efficient phase retrieval procedure, and use recent results in spectral graph theory to produce a stable performance guarantee which rivals the guarantee for PhaseLift in~\cite{CandesSV:11}.
We use numerical simulations to illustrate the performance of our phase retrieval procedure, and we compare reconstruction error and runtime with a common alternating-projections-type procedure.
\end{abstract}
}


{\small
\begin{center}
\end{center}
}

\section{Introduction}

Given a collection of vectors $\Phi:=\{\varphi_\ell\}_{\ell=1}^N\subseteq\mathbb{C}^M$ and a signal $x\in\mathbb{C}^M$, consider measurements of the form
\begin{equation}
\label{eq.phaseless measurements}
z_\ell:=|\langle x,\varphi_\ell\rangle|^2+\nu_\ell,
\end{equation}
where $\nu_\ell$ is noise; we call these noisy \emph{intensity measurements}.
Several areas of imaging science, such as X-ray crystallography~\cite{Harrison:93,Miao:08,Millane:90}, diffraction imaging~\cite{Bunk:07}, astronomy~\cite{DaintyF:87} and optics~\cite{Walther:63}, use measurements of this form with the intent of reconstructing the original signal; this inverse problem is called \emph{phase retrieval}.
Note that in the measurement process~\eqref{eq.phaseless measurements}, we inherently lose some information about $x$.
Indeed, for every $\omega\in\mathbb{C}$ with $|\omega|=1$, we see that $x$ and $\omega x$ produce the same intensity measurements.
Thus, the best one can hope to do with the intensity measurements of $x\in\mathbb{C}^M$ is reconstruct the class $[x]\in\mathbb{C}^M/\!\!\sim$, where $\sim$ is the equivalence relation of being identical up to a global phase factor.

In practice, phase retrieval falls short of determining the original signal up to global phase.
First of all, the intensity measurement process that is used, say $\mathcal{A}:\mathbb{C}^M/\!\!\sim~\!\!\!\!\rightarrow\mathbb{R}_{\geq0}^N$ with $\mathcal{A}(x)=|\Phi^* x|^2$ (entrywise) and viewing $\Phi=[\varphi_1\cdots\varphi_N]$, often lacks injectivity, making it impossible to reconstruct uniquely.
Moreover, the phase retrieval algorithms that are used in practice take alternating projections onto the column space of $\Phi^*$ (to bring phase to the measurements) and onto the nonconvex set of vectors $y$ whose entry magnitudes match the intensity measurements $|\Phi^* x|^2$ (to maintain fidelity in the magnitudes)~\cite{Fienup:82,GerchbergS:72,GriffinL:84}.
Unfortunately, the convergence of these algorithms (and various modifications thereof) is particularly sensitive to the choice of initial phases~\cite{Marchesini:07}.

These deficiencies have prompted two important lines of research in phase retrieval:
\begin{itemize}
\item[(i)] For which measurement designs $\Phi$ is $[x]\mapsto |\Phi^* x|^2$ injective?

\item[(ii)] For which injective designs can $[x]$ be reconstructed stably and efficiently?
\end{itemize}

A first step toward solving (i) is determining how large $N$ must be in order for $\mathcal{A}$ to be injective.
It remains an open problem to find the smallest such $N$~\cite{BandeiraCMN:13}, but embedding results in differential geometry give that $N\geq(4+o(1))M$ is necessary~\cite{AtiyahH:59,HeinosaariMW:12}.
As for sufficiency, Balan, Casazza and Edidin~\cite{BalanCE:06} show that for almost every choice of $\Phi$, $\mathcal{A}$ is injective whenever $N\geq 4M-2$, and recent work has leveraged more structured constructions to slightly decrease this number of measurement vectors~\cite{BodmannH:13,MondragonV:13}.
Though the community has investigated various conditions for injectivity, very little is known about how to stably and efficiently reconstruct in the injective case.
In fact, some instances of the phase retrieval problem are known to be $\mathrm{NP}$-complete~\cite{SahinoglouC:91}, and so any general reconstruction process is necessarily inefficient, assuming $\mathrm{P}\neq\mathrm{NP}$.

This leads one to attempt provably stable and efficient reconstruction from measurements of the form~\eqref{eq.phaseless measurements} with particular ensembles $\Phi$.
Until recently, this was only known to be possible in cases where $N=\Omega(M^2)$~\cite{BalanBCE:09}. 
By contrast, the state of the art comes from Cand\`{e}s, Strohmer and Voroninski~\cite{CandesSV:11}, who use semidefinite programming to stably reconstruct from $N=\mathcal{O}(M \log M)$ Gaussian-random measurements.
There is other work along this vein~\cite{CandesESV:11,CandesL:12,DemanetH:12,WaldspurgerAM:12} which also uses semidefinite programming and provides related guarantees.
Typically, semidefinite programs are solved via interior point methods.
Since these methods are computationally expensive, in practice, one is inclined to instead use faster numerical methods, but these lack performance guarantees.

While (i) and (ii) above describe what has been a more theoretical approach to phase retrieval, practitioners have meanwhile considered various alternatives to the intensity measurement process~\eqref{eq.phaseless measurements}.
A common theme among these alternatives is the use of interference to extract more information about the desired signal.
For example, \emph{holography} interferes the signal of interest $x\in\mathbb{C}^M$ with a known reference signal $y\in\mathbb{C}^M$, taking measurements of the form $|F(x+\omega y)|^2$, where $\omega\in\mathbb{C}$ has unit modulus and $F$ denotes the Fourier transform~\cite{DuadiEtal:11}; three such measurements (i.e., $3M$ scalar measurements) suffice for injectivity~\cite{Wang:13}.
Alternatively, \emph{spectral phase interferometry for direct electric-field reconstruction (SPIDER)} interferes the signal of interest $x\in\mathbb{C}^M$ with time- and frequency-shifted versions of itself $Sx\in\mathbb{C}^M$, taking measurements of the form $|F(x+Sx)|^2$~\cite{IaconisW:98}; while popular in practice for ultrashort pulse measurement, SPIDER fails to accurately resolve the relative phase of well-separated frequency components~\cite{KeustersTOZTW:03}.
Another interesting approach is \emph{ptychography}, in which overlapping spatial components $P_ix,P_jx\in\mathbb{C}^M$ are interfered with each other, and measurements have the form $|F(P_ix+P_jx)|^2$~\cite{Rodenburg:08}.
Recently, \emph{vectorial phase retrieval} was proposed, in which two unknown signals $x,y\in\mathbb{C}^M$ are interfered with each other, and the measurements are $|Fx|^2$, $|Fy|^2$, $|F(x+y)|^2$ and $|F(x+\mathrm{i}y)|^2$~\cite{Raz:11}; furthermore,~\cite{RazDN:13} gives that almost every pair of signals is uniquely determined by these $4M$ scalar measurements, in which case both signals can be reconstructed using the polarization identity.
While practitioners seem to have identified interference as an instrumental technique for quality phase retrieval, the reconstruction algorithms which are typically used, much like the classical algorithms in~\cite{Fienup:82,GerchbergS:72,GriffinL:84}, are iterative and lack convergence guarantees (e.g.,~\cite{DuadiEtal:11,MaidenR:09}, though \cite{RazDN:13,Raz:11} are noteworthy exceptions).

Returning to measurements of the form~\eqref{eq.phaseless measurements}, this paper combines ideas from both state-of-the-art theory and state-of-the-art practice by proposing an exchange of sorts:
If you already have $\mathcal{O}(M \log M)$ Gaussian-random measurements vectors (as prescribed in~\cite{CandesSV:11}), then we offer a faster reconstruction method with a stable performance guarantee, but at the price of $\mathcal{O}(M \log M)$ additional (non-adaptive) measurements.
These new measurement vectors are interferometry-inspired combinations of the originals, and the computational speedups gained in reconstruction come from our use of different spectral methods.
While the ideas in this paper can be applied for phase retrieval of 2-D images, we focus on the 1-D case for simplicity.
Also, note that the sequel~\cite{BandeiraCM:13} leverages the techniques of this paper to construct masked Fourier measurements, thereby mimicking the illumination methodology of~\cite{CandesESV:11}; we suspect that these ideas can be similarly leveraged to tackle a wide variety of practical instances of the phase retrieval problem.
To help motivate our measurement design and phase retrieval procedure, we start in the next section by considering the simpler, noiseless case.  
In this case, the success of our method follows from a neat trick involving the polarization identity along with some well-known results in the theory of expander graphs.  
In Section~3, we modify the method to obtain provable stability in the noisy case; here, we exploit some recent developments in spectral graph theory.  
Our results are then corroborated by simulations in Section~4.
We give concluding remarks in Section~5, and provide the more technical proofs in the appendix.

\section{The noiseless case}

In this section, we provide a new measurement design and reconstruction algorithm for phase retrieval.
Here, we specifically address the noiseless case, in which $\nu_\ell$ in \eqref{eq.phaseless measurements} is zero for every $\ell=1,\ldots,N$; this case will give some intuition for a more stable version of our techniques, which we introduce in the next section.
In the noiseless case, we will use on the order of the fewest measurements possible, namely $N=\mathcal{O}(M)$, where $M$ is the dimension of the signal.

Before stating our measurement design and phase retrieval procedure, we motivate both with some discussion.
Take a finite set $V$, and suppose we take intensity measurements of $x\in\mathbb{C}^M$ with a set $\Phi_V:=\{\varphi_i\}_{i\in V}$ that spans $\mathbb{C}^M$.
Again, we wish to recover $x$ up to a global phase factor.
Having $|\langle x,\varphi_i\rangle|$ for every $i\in V$, we claim it suffices to determine the relative phase between $\langle x,\varphi_i\rangle$ and $\langle x,\varphi_j\rangle$ for all pairs $i\neq j$.
Indeed, if we had this information, we could arbitrarily assign some nonzero coefficient $c_i=|\langle x,\varphi_i\rangle|$ to have positive phase.
If $\langle x,\varphi_j\rangle$ is also nonzero, then it has well-defined relative phase
\begin{equation}
\label{eq.relative phase}
\rho_{ij}:=\big(\tfrac{\langle x,\varphi_i\rangle}{|\langle x,\varphi_i\rangle|}\big)^{-1}\tfrac{\langle x,\varphi_j\rangle}{|\langle x,\varphi_j\rangle|},
\end{equation}
which determines the coefficient by multiplication: $c_j=\rho_{ij}|\langle x,\varphi_j\rangle|$.
Otherwise when $\langle x,\varphi_j\rangle=0$, we naturally take $c_j=0$, and for notational convenience, we arbitrarily take $\rho_{ij}=1$.
From here, the original signal's equivalence class $[x]\in\mathbb{C}^M/\!\!\sim$ can be identified by applying the canonical dual frame $\{\tilde{\varphi}_j\}_{j\in V}$, namely the Moore-Penrose pseudoinverse, of $\Phi_V$:
\begin{equation}
\label{eq.pseudoinverse}
\sum_{j\in V}c_j\tilde{\varphi}_j
=\sum_{j\in V}\rho_{ij}|\langle x,\varphi_j\rangle|\tilde{\varphi}_j
=\big(\tfrac{\langle x,\varphi_i\rangle}{|\langle x,\varphi_i\rangle|}\big)^{-1}\sum_{j\in V}\langle x,\varphi_j\rangle\tilde{\varphi}_j
=\big(\tfrac{\langle x,\varphi_i\rangle}{|\langle x,\varphi_i\rangle|}\big)^{-1}x
\in[x].
\end{equation}

Having established the utility of the relative phase between coefficients, we now seek some method of extracting this information.
To this end, we turn to a special version of the polarization identity:

\begin{lem}[Mercedes-Benz Polarization Identity]
Take $\zeta:=\mathrm{e}^{2\pi\mathrm{i}/3}$.  Then for any $a,b\in\mathbb{C}$,
\begin{equation}
\label{eq.pol identity}
\bar{a}b=\frac{1}{3}\sum_{k=0}^2\zeta^k|a+\zeta^{-k}b|^2.
\end{equation}
\end{lem}

\begin{proof}
We start by expanding the right-hand side of \eqref{eq.pol identity}:
\begin{equation*}
\mathrm{RHS}
:=\frac{1}{3}\sum_{k=0}^2\zeta^k|a+\zeta^{-k}b|^2
=\frac{1}{3}\sum_{k=0}^2\zeta^k\big(|a|^2+2\mathrm{Re}(\zeta^{-k}\bar{a}b)+|b|^2\big)
=\frac{2}{3}\sum_{k=0}^2\zeta^k\mathrm{Re}(\zeta^{-k}\bar{a}b).
\end{equation*}
Multiplying, we find
\begin{equation*}
\mathrm{Re}(\zeta^{-k}\bar{a}b)
=\mathrm{Re}(\zeta^{-k})\mathrm{Re}(\bar{a}b)-\mathrm{Im}(\zeta^{-k})\mathrm{Im}(\bar{a}b)
=\mathrm{Re}(\zeta^{k})\mathrm{Re}(\bar{a}b)+\mathrm{Im}(\zeta^{k})\mathrm{Im}(\bar{a}b).
\end{equation*}
We substitute this into our expression for $\mathrm{RHS}$:
\begin{align*}
\mathrm{Re}(\mathrm{RHS})
&=\frac{2}{3}\bigg[\mathrm{Re}(\bar{a}b)\sum_{k=0}^2\big(\mathrm{Re}(\zeta^k)\big)^2+\mathrm{Im}(\bar{a}b)\sum_{k=0}^2\mathrm{Re}(\zeta^k)\mathrm{Im}(\zeta^{k})\bigg],\\
\mathrm{Im}(\mathrm{RHS})
&=\frac{2}{3}\bigg[\mathrm{Re}(\bar{a}b)\sum_{k=0}^2\mathrm{Re}(\zeta^k)\mathrm{Im}(\zeta^{k})+\mathrm{Im}(\bar{a}b)\sum_{k=0}^2\big(\mathrm{Im}(\zeta^k)\big)^2\bigg].
\end{align*}
Finally, we apply the following easy-to-verify identities:
\begin{equation*}
\sum_{k=0}^2\big(\mathrm{Re}(\zeta^k)\big)^2=\sum_{k=0}^2\big(\mathrm{Im}(\zeta^k)\big)^2=\frac{3}{2},\qquad \sum_{k=0}^2\mathrm{Re}(\zeta^k)\mathrm{Im}(\zeta^{k})=0,
\end{equation*}
which yield $\mathrm{RHS}=\bar{a}b$.
\qquad
\end{proof}

The above polarization identity can also be proved by viewing $\{\zeta^k\}_{k=0}^2$ as a Mercedes-Benz frame in $\mathbb{R}^2$ and $\frac{2}{3}\sum_{k=0}^2\zeta^k\mathrm{Re}(\zeta^{-k}u)$ as the corresponding reconstruction formula for $u\in\mathbb{C}=\mathbb{R}^2$.
We can now use this polarization identity to determine relative phase \eqref{eq.relative phase}:
\begin{equation}
\label{eq.polarization trick}
\overline{\langle x,\varphi_i\rangle}\langle x,\varphi_j\rangle
=\frac{1}{3}\sum_{k=0}^2\zeta^{k}\big|\langle x,\varphi_i\rangle+\zeta^{-k}\langle x,\varphi_j\rangle\big|^2
=\frac{1}{3}\sum_{k=0}^2\zeta^{k}\big|\langle x,\varphi_i+\zeta^{k}\varphi_j\rangle\big|^2.
\end{equation}
Thus, if in addition to $\Phi_V$ we measure with $\{\varphi_i+\zeta^{k}\varphi_j\}_{k=0}^2$, we can use \eqref{eq.polarization trick} to determine $\overline{\langle x,\varphi_i\rangle}\langle x,\varphi_j\rangle$ and then normalize to get the relative phase:
\begin{equation}
\label{eq.calculate relative phase}
\rho_{ij}
:=\big(\tfrac{\langle x,\varphi_i\rangle}{|\langle x,\varphi_i\rangle|}\big)^{-1}\tfrac{\langle x,\varphi_j\rangle}{|\langle x,\varphi_j\rangle|}
=\tfrac{\overline{\langle x,\varphi_i\rangle}\langle x,\varphi_j\rangle}{|\overline{\langle x,\varphi_i\rangle}\langle x,\varphi_j\rangle|},
\end{equation}
provided both $\langle x,\varphi_i\rangle$ and $\langle x,\varphi_j\rangle$ are nonzero.
To summarize our discussion of reconstructing a single signal, if we measure with $\Phi_V$ and $\{\varphi_i+\zeta^{k}\varphi_j\}_{k=0}^2$ for every pair $i,j\in V$, then we can recover $[x]$.
However, such a method uses $|V|+3\binom{|V|}{2}$ measurements, and since $\Phi_V$ must span $\mathbb{C}^M$, we necessarily have $|V|\geq M$ and thus a total of $\Omega(M^2)$ measurements.
Note that a nearly identical formulation of these $\Omega(M^2)$ measurements appears in Theorem~5.2 of~\cite{BalanBCE:09}.
The proof of this result shows how one can adaptively appeal to only $\mathcal{O}(M)$ of the measurements to perform phase retrieval, suggesting that most of these measurements are actually unnecessary.
However, since the $\mathcal{O}(M)$ measurements that end up being used are highly dependent on the signal being measured, one cannot blindly restrict to a particular subcollection of $\mathcal{O}(M)$ measurement vectors a priori without forfeiting injectivity.

In pursuit of $\mathcal{O}(M)$ measurements, take some simple graph $G=(V,E)$, arbitrarily assign a direction to each edge, and only take measurements with $\Phi_V$ and $\Phi_E:=\bigcup_{(i,j)\in E}\{\varphi_i+\zeta^{k}\varphi_j\}_{k=0}^2$.
To recover $[x]$, we again arbitrarily assign some nonzero vertex measurement to have positive phase, and then we propagate relative phase information along the edges by multiplication to determine the phase of the other vertex measurements relative to the original vertex measurement:
\begin{equation}
\label{eq.propagate relative phase}
\rho_{ik}
=\rho_{ij}\rho_{jk}.
\end{equation}
However, if $x$ is orthogonal to a given vertex vector, then that measurement is zero, and so relative phase information cannot propagate through the corresponding vertex; indeed, such orthogonality has the effect of removing the vertex from the graph, and for some graphs, this will prevent recovery.
For example, if $G$ is a star, then $x$ could be orthogonal to the vector corresponding to the internal vertex, whose removal would render the remaining graph edgeless.
That said, we should select $\Phi_V$ and $G$ so as to minimize the impact of orthogonality with vertex vectors.

First, we can take $\Phi_V$ to be \emph{full spark}, that is, $\Phi_V$ has the property that every subcollection of $M$ vectors spans.
Full spark frames appear in a wide variety of applications.
Explicit deterministic constructions of them are given in~\cite{AlexeevCM:12,PuschelK:05}.
For example, we can select the first $M$ rows of the $|V|\times|V|$ discrete Fourier transform matrix, and take $\Phi_V$ to be the columns of the resulting $M\times|V|$ matrix; in this case, the fact that $\Phi_V$ is full spark follows from the Vandermonde determinant formula.
In our application, $\Phi_V$ being full spark will be useful for two reasons.
First, this implies that $x\neq 0$ is orthogonal to at most $M-1$ members of $\Phi_V$, thereby limiting the extent of $x$'s damage to our graph.
Additionally, $\Phi_V$ being full spark frees us from requiring the graph to be connected after the removal of vertices; indeed, any remaining component of size $M$ or more will correspond to a subcollection of $\Phi_V$ that spans, meaning it has a dual frame to reconstruct with.
It remains to find a graph of $\mathcal{O}(M)$ vertices and edges that maintains a size-$M$ component after the removal of any $M-1$ vertices.

To this end, we consider a well-studied family of sparse graphs known as \emph{expander graphs}.
We choose these graphs for their notably strong connectivity properties.
There is a combinatorial definition of expander graphs, but we will focus on the spectral definition.
Given a $d$-regular graph $G$ of $n$ vertices, consider its adjacency matrix $A$, and define the \emph{Laplacian} to be $L:=I-\frac{1}{d}A$; if $G$ were not regular, we would consider the diagonal matrix $D$ of vertex degrees and define the Laplacian to be $L:=I-D^{-1/2}AD^{-1/2}$.
This is often called the \textit{normalized} Laplacian in the literature, but we make no distinction here.
We are particularly interested in the eigenvalues of the Laplacian: $0=\lambda_1\leq\cdots\leq\lambda_n$.
The second eigenvalue $\lambda_2$ of the Laplacian is called the \emph{spectral gap} of the graph, and as we shall see, this value is particularly useful in evaluating the graph's connectivity.
We say $G$ has \emph{expansion} $\lambda$ if $\{\lambda_2,\ldots,\lambda_n\}\subseteq[1-\lambda,1+\lambda]$; note that since $1-\lambda\leq\lambda_2$, small expansion implies large spectral gap.
Furthermore, a family of $d$-regular graphs $\{G_i\}_{i=1}^\infty$ is a \emph{spectral expander family} if there exists $c<1$ such that every $G_i$ has expansion $\lambda(G_i)\leq c$.
Since $d$ is constant over an expander family, expanders with many vertices have particularly few edges.
There are many results which describe the connectivity of expanders, but the following is particularly relevant to our application:

\begin{lem}[Spectral gap grants connectivity~\cite{HarshaB:online}]\label{Lemma:noiselesstrimming}
Consider a $d$-regular graph $G$ of $n$ vertices with spectral gap $\lambda_2$.
For all $\varepsilon\leq\frac{\lambda_2}{6}$, removing any $\varepsilon dn$ edges from $G$ results in a connected component of size $\geq(1-\frac{2\varepsilon}{\lambda_2})n$.
\end{lem}

Note that removing $\varepsilon n$ vertices from a $d$-regular graph necessarily removes $\leq\varepsilon dn$ edges, and so this lemma directly applies.
For our application, we want to guarantee that the removal of any $M-1$ vertices maintains a size-$M$ component.
To do this, we will ensure both (i) $M-1\leq\varepsilon n$ and (ii) $M-1<(1-\frac{2\varepsilon}{\lambda_2})n$, and then invoke the above lemma.
Note that since $n\geq M\geq2$,
\begin{equation*}
\varepsilon
\leq\frac{\lambda_2}{6}
\leq\frac{\mathrm{Tr}[L]}{6(n-1)}
=\frac{n}{6(n-1)}
\leq\frac{1}{3}
<\frac{2}{3}
\leq 1-\frac{2\varepsilon}{\lambda_2},
\end{equation*}
where the last inequality is a rearrangement of $\varepsilon\leq\frac{\lambda_2}{6}$.
Thus $\varepsilon n<(1-\frac{2\varepsilon}{\lambda_2})n$, meaning (i) implies (ii), and so it suffices to have $M\leq\varepsilon n+1$.
Overall, we use the following criteria to pick our expander graph:
Given the signal dimension $M$, use a $d$-regular graph $G=(V,E)$ of $n$ vertices with spectral gap $\lambda_2$ such that $M\leq(\frac{\lambda_2}{6})n+1$.
Then by the previous discussion, the total number of measurements is $N=|V|+3|E|=(\frac{3}{2}d+1)n$.
If we think of the degree $d$ as being fixed, then the number of vertices $n$ in the graph is proportional to the total number of measurements $N$ (this is the key distinction from the previous complete-graph case).

Recall that we seek $N=\mathcal{O}(M)$ measurements.
To minimize the redundancy $\frac{N}{M}$ for a fixed degree $d$, we would like a maximal spectral gap $\lambda_2$, and it suffices to seek minimal spectral expansion $\lambda$.
Spectral graph families known as \emph{Ramanujan graphs} are asymptotically optimal in this sense; taking $\mathcal{G}_n^d$ to be the set of connected $d$-regular graphs with $\geq n$ vertices, Alon and Boppana (see~\cite{Alon:86}) showed that for any fixed $d$,
\begin{equation*} \lim_{n\rightarrow\infty}\inf_{G\in\mathcal{G}_n^d}\lambda(G)\geq\frac{2\sqrt{d-1}}{d},
\end{equation*}
while Ramanujan graphs are defined to have spectral expansion $\leq\frac{2\sqrt{d-1}}{d}$.
To date, Ramanujan graphs have only been constructed for certain values of $d$.
One important construction was given by Lubotzky, Phillips, and Sarnak~\cite{LubotzkyPS:88}, which produces a Ramanujan family whenever $d-1\equiv 1\bmod4$ is prime.
Among these graphs, we get the smallest redundancy $\frac{N}{M}$ when $M=\lfloor(1-\frac{2\sqrt{d-1}}{d})\frac{n}{6}+1\rfloor$ and $d=6$:
\begin{equation*}
\frac{N}{M}
\leq\frac{(\frac{3}{2}d+1)n}{(1-\frac{2\sqrt{d-1}}{d})\frac{n}{6}}
=45\big(3+\sqrt{5}\big)
\approx 235.62.
\end{equation*}
Thus, in such cases, our techniques allow for phase retrieval with only $N\leq 236M$ measurements.
However, the number of vertices in each Ramanujan graph from~\cite{LubotzkyPS:88} is of the form $q(q^2-1)$ or $\frac{q(q^2-1)}{2}$, where $q\equiv1\bmod4$ is prime, and so any bound on redundancy $\frac{N}{M}$ using these graphs will only be valid for particular values of $M$.

In order to get $N=\mathcal{O}(M)$ in general, we use the fact that random graphs are nearly Ramanujan with high probability.
In particular, for every $\varepsilon>0$ and even $d$, a random $d$-regular graph has spectral expansion $\lambda\leq\frac{2\sqrt{d-1}+\varepsilon}{d}$ with high probability as $n\rightarrow\infty$~\cite{Friedman:08}.
Thus, picking $\varepsilon$ and $d$ to satisfy $\frac{2\sqrt{d-1}+\varepsilon}{d}<1$, we may take $M=\lfloor(1-\frac{2\sqrt{d-1}+\varepsilon}{d})\frac{n}{6}+1\rfloor$ to get
\begin{equation*}
\frac{N}{M}
\leq\frac{(\frac{3}{2}d+1)n}{(1-\frac{2\sqrt{d-1}+\varepsilon}{d})\frac{n}{6}},
\end{equation*}
and this choice will satisfy $M\leq(\frac{\lambda_2}{6})n+1$ with high probability.
To see how small this redundancy is, note that taking $\varepsilon=0.1$ and $d=8$ gives $N\leq 240 M$.
While the desired expansion properties of a random graph are only present with high probability, estimating the spectral gap is inexpensive, and so it is computationally feasible to verify whether a randomly drawn graph is good enough.
Moreover, $n$ can be any sufficiently large integer, and so the above bound is valid for all sufficiently large $M$, i.e., our procedure can perform phase retrieval with $N=\mathcal{O}(M)$ measurements in general.

Combining this with the above discussion, we have the following measurement design and phase retrieval procedure:

\medskip
\noindent\textbf{Measurement Design~A (noiseless case)}
\begin{itemize}
\item Fix $d>2$ even and $\varepsilon\in(0,d-2\sqrt{d-1})$.
\item Given $M$, pick some $d$-regular graph $G=(V,E)$ with spectral gap $\lambda_2\geq\lambda':=1-\frac{2\sqrt{d-1}+\varepsilon}{d}$ and $|V|=\lceil \frac{6}{\lambda'}(M-1)\rceil$, and arbitrarily direct the edges.
\item Design the measurements $\Phi:=\Phi_V\cup\Phi_E$ by taking $\Phi_V:=\{\varphi_i\}_{i\in V}\subseteq\mathbb{C}^M$ to be full spark and $\Phi_E:=\bigcup_{(i,j)\in E}\{\varphi_i+\zeta^{k}\varphi_j\}_{k=0}^2$.
\end{itemize}
\medskip
\noindent\textbf{Phase Retrieval Procedure~A (noiseless case)}
\begin{itemize}
\item Given $\{|\langle x,\varphi\rangle|^2\}_{\varphi\in\Phi}$, delete the vertices $i\in V$ with $|\langle x,\varphi_i\rangle|^2=0$.
\item In the remaining induced subgraph, find a connected component of $\geq M$ vertices $V'$.
\item Pick a vertex in $V'$ to have positive phase and propagate/multiply relative phases \eqref{eq.propagate relative phase}, which are calculated by normalizing \eqref{eq.polarization trick}, see \eqref{eq.calculate relative phase}.
\item Having $\{\langle x,\varphi_i\rangle\}_{i\in V'}$ up to a global phase factor, find the least-squares estimate of~$[x]$ by applying the Moore-Penrose pseudoinverse of $\{\varphi_i\}_{i\in V'}$, see \eqref{eq.pseudoinverse}.
\end{itemize}
\medskip

Note that this phase retrieval procedure is particularly fast.
Indeed, if we use $E\subseteq V^2$ to store $G$, then we can delete vertices $i\in V$ with $|\langle x,\varphi_i\rangle|^2=0$ by deleting the edges for which \eqref{eq.polarization trick} is zero, which takes $\mathcal{O}(|E|)$ time.
Next, if the members of $E$ are ordered lexicographically, the remaining subgraph can be easily partitioned into connected components in $\mathcal{O}(|E|)$ time by collecting edges with common vertices, and then propagating relative phase in the largest component is performed in $\mathcal{O}(|E|)$ time using a depth- or breadth-first search.
Overall, we only use $\mathcal{O}(M)$ time before the final least-squares step of the phase retrieval procedure, which happens to be the bottleneck, depending on the subcollection $\Phi_{V'}$.
In general, we can find the least-squares estimate in $\mathcal{O}(M^3)$ time using Gaussian elimination, but if $\Phi_{V'}$ has special structure (e.g., it is a submatrix of the discrete Fourier transform matrix), then one might exploit that structure to gain speedups (e.g., use the fast Fourier transform in conjunction with an iterative method).
Regardless, our procedure reduces the nonlinear phase retrieval problem to the much simpler problem of solving an overdetermined \emph{linear} system.

While this measurement design and phase retrieval procedure is particularly efficient, it certainly lacks stability.
Perhaps most notably, we have not imposed anything on~$\Phi_V$ that guarantees stability with inverting~$\Phi_{V'}$; indeed, we have merely enforced linear independence between vectors, while stability will require well-conditioning.
Another noteworthy source of instability is our method of phase propagation, which naturally accumulates error; it would be better if the relative phases were combined using a more democratic process that encourages noise cancellation.
In the next section, we will address these concerns (and others) and modify our procedure accordingly; the revised procedure will be stable, but at the price of a log factor in the number of measurements: $N=\mathcal{O}(M\log M)$.
As we mention in the concluding remarks, we do not think this log factor is necessary, but we leave this pursuit for future work.

\section{The noisy case}
\label{section.noisy}

In this section, we consider a noise-robust version of the measurement design and phase retrieval procedure of the previous section.  In the end, the measurement design will be nearly identical: vertex measurements will be independent complex Gaussian vectors (thereby being full spark with probability 1), and the edge measurements will be the same sort of linear combinations of vertex measurements.  Our use of randomness in this version will enable the vertex measurements to simultaneously satisfy two important conditions with high probability: \emph{projective uniformity with noise} and \emph{numerical erasure robustness}.  Before defining these conditions, we motivate them by considering a noisy version of our phase retrieval procedure.

Recall that our noiseless procedure starts by removing the vertices $i\in V$ for which $|\langle x,\varphi_i\rangle|^2=0$.  Indeed, since we plan to propagate relative phase information along edges, these $0$-vertices are of no use, as relative phase with these vertices is not well defined.  Since we calculate relative phase by normalizing~\eqref{eq.polarization trick}, we see that relative phase is sensitive to perturbations when~\eqref{eq.polarization trick} is small, meaning either $\langle x,\varphi_i\rangle$ or $\langle x,\varphi_j\rangle$ is small.  As such, while $0$-vertices provide no relative phase information in the noiseless case, small vertices provide \emph{unreliable} information in the noisy case, and so we wish to remove them accordingly (alternatively, one might use weights according to one's confidence in the information, but we decided to use hard thresholds to simplify the analysis).  However, we also want to ensure that there are only a few small vertices.  In the noiseless case, we limit the number of $0$-vertices by using a full spark frame; in the noisy case, we make use of a new concept we call \emph{projective uniformity}:

\begin{definition}
\label{definition:PhaselessRadialSpreadConst_testing}
The $\alpha$-\textit{projective uniformity} of $\Phi=\{\varphi_i\}_{i=1}^n\subseteq\mathbb{C}^M$ is given by
\begin{equation*}
\pu(\Phi;\alpha)
= \min_{\substack{x\in \CC^M\\\|x\|=1}} \max_{\substack{\III\subseteq \{1,\ldots,n\}\\ |\III| \geq \alpha n}} \min_{i\in\III}|\langle x,\varphi_i \rangle|^2.
\end{equation*}
\end{definition}

In words, projective uniformity gives the following guarantee:  For every unit-norm signal $x$, there exists a collection of vertices $\mathcal{I}\subseteq V$ of size at least $\alpha|V|$ such that $|\langle x,\varphi_i\rangle|^2\geq\mathrm{PU}(\Phi_V;\alpha)$ for every $i\in\mathcal{I}$.  As such, projective uniformity effectively limits the total number of small vertices possible, at least before the measurements are corrupted by noise.  However, the phase retrieval algorithm will only have access to noisy versions of the measurements, and so we must account for this subtlety in our procedure.  In an effort to isolate the reliable pieces of relative phase information, Algorithm~\ref{alg:takingedgestotakevert} removes the vertices corresponding to small noisy edge combinations~\eqref{eq.polarization trick}.

\begin{algorithm}[h]
\caption{Pruning for reliability\label{alg:takingedgestotakevert}}
\SetAlgoLined
\KwIn{Graph $G=(V,E)$, function $f\colon E\to\RR$ such that $f(i,j) = |\overline{\langle x,\varphi_i\rangle}\langle x,\varphi_j\rangle+\varepsilon_{ij}|$, pruning parameter $\alpha$}
\KwOut{Subgraph $H$ with a larger smallest edge weight}
Initialize $H\leftarrow G$\\
\For{$i=1$ \KwTo $\lfloor(1-\alpha)|V|\rfloor$}{
Find the minimizer $(i,j)\in E$ of $f$\\
$H\leftarrow H\setminus\{i,j\}$
}
\end{algorithm}

We now explain why only reliable pieces of relative phase information will remain after running the above algorithm, provided $\Phi_V$ has sufficient projective uniformity.  The main idea is captured in the following:

\begin{lemma}\label{proposition:additiveVsangularNOISE}
Define $\|\theta\|_\mathbb{T}:=\min_{k\in\mathbb{Z}}|\theta-2\pi k|$ for all angles $\theta\in\mathbb{R}/2\pi\mathbb{Z}$.
Then for any $z,\varepsilon\in\mathbb{C}$,
\begin{equation*}
\|\arg(z+\varepsilon)-\arg(z)\|_\mathbb{T}\leq\pi\frac{|\varepsilon|}{|z|}.
\end{equation*}
\end{lemma}

\begin{proof}
If $|\varepsilon|\geq|z|$, the result is trivial.
Suppose $|\varepsilon|<|z|$, and consider the triangle whose vertices are $0$, $z$, and $z+\varepsilon$.
Denoting the angle at $0$ by $\theta=\|\arg(z+\varepsilon)-\arg(z)\|_\mathbb{T}$ and the angle at $z+\varepsilon$ by $\alpha$, the law of sines gives
\begin{equation}
\label{eq.sine bound}
\frac{\sin \theta}{|\varepsilon|}=\frac{\sin\alpha}{|z|}\leq\frac{1}{|z|}.
\end{equation}
Next, the law of cosines gives
\begin{equation*}
2|z||z+\varepsilon|\cos\theta
=|z|^2+|z+\varepsilon|^2-|\varepsilon|^2
>|z+\varepsilon|^2
>0,
\end{equation*}
and so $\cos\theta>0$, i.e., $\theta\in[0,\frac{\pi}{2})$.
Finally, by the concavity of $\sin(\cdot)$ and then \eqref{eq.sine bound}, we conclude that
\begin{equation*}
\frac{2}{\pi}\theta
\leq\sin\theta
\leq\frac{|\varepsilon|}{|z|},
\end{equation*}
which implies the result.
\qquad
\end{proof}

By taking $z=\overline{\langle x,\varphi_i\rangle}\langle x,\varphi_j\rangle+\varepsilon_{ij}$ and $\varepsilon=-\varepsilon_{ij}$, we can use this lemma to bound the relative phase error we incur when normalizing $z$.  In fact, consider the minimum of $f$ when Algorithm~\ref{alg:takingedgestotakevert} is complete.  Since the algorithm deletes vertices from $G$ according to the input signal $x$, this minimum will vary with $x$; let $\mathrm{PUN}$ denote the smallest possible minimum value.  Then the relative phase error incurred with $(i,j)\in E$ is no more than $\pi|\varepsilon_{ij}|/\mathrm{PUN}$, regardless of the signal measured.  Indeed, our use of \textit{projective uniformity with noise} (i.e., $\mathrm{PUN}$) is intended to bound the instability that comes with normalizing small values of~\eqref{eq.polarization trick}.  As we will show in the appendix, $\mathrm{PUN}$ can be bounded below by using the projective uniformity of $\Phi_V$, and furthermore, a complex Gaussian $\Phi_V$ has projective uniformity with overwhelming probability.

After applying Algorithm~\ref{alg:takingedgestotakevert}, our graph will have slightly fewer vertices, but the remaining edges will correspond to reliable pieces of relative phase information.  Recall that we plan to use this information on the edges to determine phases for the vertices, and we want to do this in a stable way.  To understand when this is even possible, we first consider a few simple scenarios.  Suppose that after removing vertices with Algorithm~\ref{alg:takingedgestotakevert}, the graph has a vertex of degree 0.  Then we have no information about the phase of this vertex, and it should be removed accordingly.  For a less extreme scenario, suppose the vertex has degree 1.  Then any noise in the corresponding edge measurement would be passed directly to the vertex, which inherently lacks stability compared to the noise cancellation that would come with more edges.  More generally, if the graph has a cut vertex (e.g., the neighbor of a degree-1 vertex), then we would need to rely on the correctness of this lone vertex to ensure consistency between the parts of the graph it connects---this scenario is also rather unstable.  After considering these examples, it makes intuitive sense that stability necessitates a high level of connectivity in the graph, regardless of the algorithm used to extrapolate the vertex phases.  

As such, we seek to remove a small proportion of vertices so that the remaining graph is very connected, i.e., has large spectral gap.
To do this, we will iteratively remove sets of vertices that are poorly connected to the rest of the graph.
These sets will be identified using \emph{spectral clustering} (Algorithm~\ref{alg:find_spectralclustering}), a process which is strongly motivated by an inequality in Riemannian geometry by Cheeger~\cite{Cheeger:70} and which has performance guarantees originating with Alon~\cite{Alon:86,AlonM:85}.
The main idea of spectral clustering follows the intuition that a random walk on a graph tends to be trapped in sections of the graph which have few connections to the rest of the vertices (this intuition is made more explicit in~\cite{MailaS:01,NadlerLCK:05}).
Moreover, the second eigenvector of the corresponding stochastic matrix tends to identify these sections.

\begin{algorithm}[h]
\caption{Spectral clustering\label{alg:find_spectralclustering}}
\SetAlgoLined
\KwIn{Graph $G=(V,E)$}
\KwOut{Subset of vertices $S$}
Take $D$ to be the diagonal matrix of vertex degrees\\
Compute the Laplacian $L\leftarrow I-D^{-1/2}AD^{-1/2}$\\
Compute the eigenvector $u$ corresponding to the second eigenvalue of $L$\\
\For{$i=1$ \KwTo $|V|$}{
Let $S_i$ denote the vertices corresponding to the $i$ smallest entries of $D^{-1/2}u$\\
Let $E(S_i,S_i^\mathrm{c})$ denote the number of edges between $S_i$ and $S_i^\mathrm{c}$\\
$h_i\leftarrow E(S_i,S_i^\mathrm{c})/\min\{\sum_{v\in S_i}\operatorname{deg}(v),\sum_{v\in S_i^\mathrm{c}}\operatorname{deg}(v)\}$
}
Take $S$ to be the $S_i$ of minimal $h_i$ (or $S_i^\mathrm{c}$ if this has smaller size)
\end{algorithm}

Note that when implementing spectral clustering, the values of $E(S_i,S_i^\mathrm{c})$, $\vol(S_i)$ and $\vol(S_i^\mathrm{c})$ can be cheaply updated from their $(i-1)$st values since $S_i$ differs from $S_{i-1}$ in only one vertex.
As such, the bottleneck in spectral clustering is computing an eigenvector.
For our application, we will iteratively apply spectral clustering to identify small collections of vertices which are poorly connected to the rest of the graph and then remove them to enhance connectivity (Algorithm~\ref{alg:find_subexpander}).

\begin{algorithm}[h]
\caption{Pruning for connectivity\label{alg:find_subexpander}}
\SetAlgoLined
\KwIn{Graph $G=(V,E)$, pruning parameter $\tau$}
\KwOut{Subgraph $H$ with spectral gap $\lambda_2(H)\geq\tau$}
Initialize $H\leftarrow G$\\
\While{$\lambda_2(H)<\tau$}{
Perform spectral clustering (Algorithm \ref{alg:find_spectralclustering}) to identify a small set of vertices $S$\\
$H\leftarrow H\setminus S$
}
\end{algorithm}

In the appendix, we show that for a particular choice of threshold $\tau$, Algorithm~\ref{alg:find_subexpander} recovers a level of connectivity that may have been lost when pruning for reliability in Algorithm~\ref{alg:takingedgestotakevert}, and it does so by removing only a small proportion of the vertices.

At this point, we have pruned our graph so that the measured relative phases are reliable and the vertex phases can be stably reconstructed.  Now we seek an efficient method to reconstruct these vertex phases from the measured relative phases.  Before devising such a method, we first organize the information we have into a matrix.  Given the graph output $G'=(V',E')$ of Algorithm~\ref{alg:find_subexpander}, we take $A_1$ to be a $|V'|\times|V'|$ weighted adjacency matrix of $G'$.
Specifically, for each $\{i,j\}\in E'$, let $\varepsilon_{ij}$ denote the effective noise in the estimate of $\overline{\langle x,\varphi_i\rangle}\langle x,\varphi_j\rangle$ using \eqref{eq.polarization trick}, and normalize this noisy estimate to get
\begin{equation}
\label{eq.weighted adjacency}
A_1[i,j]=\frac{\overline{\langle x,\varphi_i\rangle}\langle x,\varphi_j\rangle+\varepsilon_{ij}}{|\overline{\langle x,\varphi_i\rangle}\langle x,\varphi_j\rangle+\varepsilon_{ij}|}.
\end{equation}
Otherwise when $\{i,j\}\not\in E'$, take $A_1[i,j]=0$.  Unlike the noiseless case, here, we account for both directions $(i,j)$ and $(j,i)$ whenever $\{i,j\}\in E'$, with the understanding that $\varepsilon_{ji}=\overline{\varepsilon_{ij}}$; this will simplify our analysis since this makes $A_1$ self-adjoint.  Considering $A_1[i,j]$ is an approximation of the relative phase $\omega_i^{-1}\omega_j$, it seems reasonable to extrapolate the vertex phases $\omega:=\{\omega_i\}_{i\in V'}$ from $A_1$ by minimizing the following quantity:
\begin{equation*}
\sum_{\{i,j\}\in E'}|\omega_j-A_1[i,j]\omega_i|^2
=\sum_{\{i,j\}\in E'}\Big(|\omega_j|^2-2~\mathrm{Re}~\overline{\omega_j}A_1[i,j]\omega_i+|A_1[i,j]\omega_i|^2\Big)
=\omega^*(D-A_1)\omega,
\end{equation*}
where $D$ is the diagonal matrix of vertex degrees.
Dividing by $\vol(G')=\omega^*D\omega$, which does not vary with $\omega$, this is equivalent to minimizing
\begin{equation*}
\frac{\omega^*(D-A_1)\omega}{\omega^*D\omega}
=\frac{(D^{1/2}\omega)^*(I-D^{-1/2}A_1D^{-1/2})(D^{1/2}\omega)}{\|D^{1/2}\omega\|^2}
\geq\lambda_1(I-D^{-1/2}A_1D^{-1/2}).
\end{equation*}
To be clear, the right-hand side above is the first eigenvalue of $L_1:=I-D^{-1/2}A_1D^{-1/2}$, which we call the \textit{connection Laplacian}; note that this bears some resemblance to the Laplacian defined in the previous section.
In minimizing the above quantity, it makes sense to consider the eigenvector $u$ corresponding to the smallest eigenvalue of $L_1$, but we require each coordinate of $D^{-1/2}u$ to have unit modulus. 
Provided $u$ has no entries which are zero, we can normalize the entries to form an estimate of $\omega$, and as we show in the appendix (using results from~\cite{BandeiraSS:12}), this estimate is stable provided the spectral gap of $G'$ is sufficiently large.  
This spectral method is known in the literature as \textit{angular synchronization}~\cite{Singer:11}, and we summarize the procedure in Algorithm~\ref{alg:ang_synch_appendix}

\begin{algorithm}[h]
\caption{Angular synchronization\label{alg:ang_synch_appendix}}
\SetAlgoLined
\KwIn{Graph $G'=(V',E')$, noisy versions of \eqref{eq.polarization trick} for every $\{i,j\}\in E'$}
\KwOut{Vector of phases corresponding to vertex measurements}
Let $A_1$ denote the matrix given by \eqref{eq.weighted adjacency} whenever $\{i,j\}\in E'$, and otherwise $A_1[i,j]=0$\\
Let $D$ denote the diagonal matrix of vertex degrees\\
Compute the connection Laplacian $L_1\leftarrow I-D^{-1/2}A_1D^{-1/2}$\\
Compute the eigenvector $u$ corresponding to the smallest eigenvalue of $L_1$\\
Output the phases of the coordinates of $u$
\end{algorithm}

To reiterate, Algorithm~\ref{alg:ang_synch_appendix} will produce estimates for the phases of the inner products $\{\langle x,\varphi_i\rangle\}_{i\in V'}$.
Also, we can take square roots of the vertex measurements $\{|\langle x,\varphi_i\rangle|^2+\nu_i\}_{i\in V'}$ to estimate $\{|\langle x,\varphi_i\rangle|\}_{i\in V'}$.
Then we can combine these to estimate $\{\langle x,\varphi_i\rangle\}_{i\in V'}$.
However, note that the largest of these inner products will be most susceptible to noise in the corresponding phase estimate.
As such, we remove a small fraction of these largest vertices so that the final collection of vertices $V''$ has size $\kappa|V|$, where $V$ was the original vertex set, and $\kappa$ is sufficiently close to $1$.

Now that we have estimated the phases of $\{\langle x,\varphi_i\rangle\}_{i\in V''}$, we wish to reconstruct $x$ by applying the Moore-Penrose pseudoinverse of $\{\varphi_i\}_{i\in V''}$.  However, since $V''$ is likely a strict subset of $V$, it can be difficult in general to predict how stable the pseudoinverse will be.  Fortunately, a recent theory of \textit{numerically erasure-robust frames (NERFs)} makes this prediction possible:  If the members of $\Phi_V$ are independent Gaussian vectors, then with high probability, every submatrix of columns $\Phi_{V''}$ with $\kappa=|V''|/|V|$ sufficiently large has a stable pseudoinverse~\cite{FickusM:12}.  This concludes the phase retrieval procedure, briefly outlined below together with the measurement design.

\medskip
\noindent\textbf{Measurement Design~B (noisy case)}
\begin{itemize}
\item Fix $d>2$ even and $\varepsilon\in(0,d-2\sqrt{d-1})$.
\item Given $M$, pick some $d$-regular graph $G=(V,E)$ with spectral gap $\lambda_2\geq\lambda':=1-\frac{2\sqrt{d-1}+\varepsilon}{d}$ and $|V|=cM\log M$ for $c$ sufficiently large, and arbitrarily direct the edges.
\item Design the measurements $\Phi:=\Phi_V\cup\Phi_E$ by taking $\Phi_V:=\{\varphi_i\}_{i\in V}\subseteq\mathbb{C}^M$ to have independent entries with distribution $\mathbb{C}\mathcal{N}(0,\frac{1}{M})$ and $\Phi_E:=\bigcup_{(i,j)\in E}\{\varphi_i+\zeta^{k}\varphi_j\}_{k=0}^2$.
\end{itemize}
\medskip
\noindent\textbf{Phase Retrieval Procedure~B (noisy case)}
\begin{itemize}
\item Given $\{|\langle x,\varphi_\ell\rangle|^2+\nu_\ell\}_{\ell=1}^N$, prune the graph $G$, keeping only reliable vertices (Algorithm~\ref{alg:takingedgestotakevert}).
\item Prune the remaining induced subgraph for connectivity, producing the vertex set $V'$ (Algorithm~\ref{alg:find_subexpander}).
\item Estimate the phases of the vertex measurements using angular synchronization (Algorithm~\ref{alg:ang_synch_appendix}).
\item Remove the vertices with the largest measurements, keeping only $|V''|=\kappa|V|$.
\item Having estimates for $\{\langle x,\varphi_i\rangle\}_{i\in V''}$ up to a global phase factor, find the least-squares estimate of~$[x]$ by applying the Moore-Penrose pseudoinverse of $\{\varphi_i\}_{i\in V''}$, see \eqref{eq.pseudoinverse}.
\end{itemize}
\medskip

Having established our measurement design and phase retrieval procedure for the noisy case, we now present the following guarantee of stable performance:

\begin{theorem}
\label{thm.main result}
Pick $N\sim CM\log M$ with $C$ sufficiently large, and take $\{\varphi_\ell\}_{\ell=1}^N=\Phi_V\cup\Phi_E$ defined in Measurement Design~B.
Then there exist constants $C',K>0$ such that the following guarantee holds for all $x\in\mathbb{C}^M$ with overwhelming probability:
Consider measurements of the form
\begin{equation*}
z_\ell:=|\langle x,\varphi_\ell\rangle|^2+\nu_\ell.
\end{equation*}
If the noise-to-signal ratio satisfies $\mathrm{NSR}:=\frac{\|\nu\|}{\|x\|^2}\leq\frac{C'}{\sqrt{M}}$, then Phase Retrieval Procedure~B produces an estimate $\tilde{x}$ from $\{z_\ell\}_{\ell=1}^N$ with squared relative error
\begin{equation*}
\frac{\|\tilde{x}-\mathrm{e}^{\mathrm{i}\theta}x\|^2}{\|x\|^2}
\leq K\sqrt{\frac{M}{\log M}}~\mathrm{NSR}
\end{equation*}
for some phase $\theta\in [0,2\pi)$.
\end{theorem}

The interested reader is directed to the appendix for a proof of this guarantee.
Before concluding this section, we evaluate the result.
Note that the norms of the $\varphi_\ell$'s tend to be $\mathcal{O}(1)$, and so the noiseless measurements $|\langle x,\varphi_\ell\rangle|^2$ tend to be of size $\mathcal{O}(\|x\|^2/M)$.
Also, in the worst-case scenario, the noise annihilates our measurements $\nu_\ell=-|\langle x,\varphi_\ell\rangle|^2$, rendering the signal $x$ unrecoverable; in this case, $\|\nu\|=\mathcal{O}(\|x\|^2\sqrt{(\log M)/M})$ since $N=CM\log M$.
In other words, if we allowed the noise-to-signal ratio to scale slightly larger than $C'/\sqrt{M}$ (i.e., by a log factor), then it would be impossible to perform phase retrieval in the worst case.
As such, the above guarantee is optimal in some sense.
Furthermore, since $\sqrt{M/\log M}~\mathrm{NSR}=\mathcal{O}(1/\sqrt{\log M})$ by assumption, the result indicates that our phase retrieval process exhibits more stability as $M$ grows large.

In comparison to the state of the art, namely, the work of Cand\`{e}s, Strohmer and Voroninski~\cite{CandesSV:11}, the most visible difference between our stability results is how we choose to scale the measurement vectors.
Indeed, the measurement vectors of the present paper tend to have norm $\mathcal{O}(1)$, whereas the measurement vectors of Cand\`{e}s et al.\ are all scaled to have norm $\sqrt{M}$; considering the statement of Theorem~\ref{thm.main result} is riddled with square roots of $M$, either choice of scaling is arguably natural.
For the sake of a more substantial comparison, their result (Theorem~1.2) gives that if the $\varphi_\ell$'s are uniformly sampled from the sphere of radius $\sqrt{M}$, and if $\|\nu\|\leq\epsilon$, then with high probability, there exists $C_0>0$ such that
\begin{equation*}
\frac{\|\tilde{x}-\mathrm{e}^{\mathrm{i}\theta}x\|}{\|x\|}
\leq C_0\min\bigg\{1,\frac{\epsilon}{\|x\|^2}\bigg\}
\end{equation*}
for some phase $\theta\in [0,2\pi)$; here, $\tilde{x}$ is the estimate which comes from PhaseLift.
(This result actually suffers from the subtlety that it does not hold for all signals $x\in\mathbb{C}^M$ simultaneously, but this was later rectified in a sequel~\cite{CandesL:12}.)
Note that the $1$ in the minimum takes effect when $\|x\|^2<\epsilon$, meaning it is possible that $\nu_\ell=-|\langle x,\varphi_\ell\rangle|^2$ (corresponding to our worst-case scenario, above); as such, this part of the guarantee is not as interesting.
The other part of the guarantee is particularly interesting:
Ignoring the $\sqrt{M/\log M}$ factor in Theorem~\ref{thm.main result} and identifying $\epsilon/\|x\|^2$ with NSR, we see that the main difference between the two guarantees is that Cand\`{e}s et al.\ bound relative error in terms of NSR rather than bounding \emph{squared} relative error.
In this sense, their result is stronger.
On the other hand, they take $\epsilon$ as an input to their reconstruction algorithm (PhaseLift), whereas our method is agnostic to the size of $\nu$.
In particular, our guarantee is continuous in the sense that relative error vanishes with $\nu$. 

\section{Numerical results}

In the previous sections, we described measurement designs and phase retrieval procedures for both the noiseless and noisy cases.
This section presents results from numerical simulations to illustrate how well our phase retrieval procedures perform in practice.
In particular, we will consider the noiseless and noisy cases separately.

\subsection{The noiseless case}

For this case, we consider a slightly different measurement design.
Rather than drawing a $d$-regular graph of $n$ vertices at random, we instead draw an Erd\H{o}s-R\'{e}nyi random graph.
That is, for a fixed $n$ and $c\leq n$, we take $n$ vertices, and place an edge between each pair of vertices independently with probability $p:=c/n$.
Note that for this model, the mean degree of each vertex is $p(n-1)=c(n-1)/n\approx c$.
As we will see, this slight change to the graph model will not adversely affect the quality of our methods.

In the simulation, we considered three values for the dimension $M\in\{16,32,64\}$.
For each of these dimensions, we let $n$ range from $M$ to $4M$ in increments of $M/4$.
In MATLAB\textregistered\ parlance, denoting $r:=n/M$, we took $r=1:0.25:4$.
We also took $c=1:0.5:2$.
For each $(r,c)$ pair, we constructed an Erd\H{o}s-R\'{e}nyi random graph $G=(V,E)$ with $n=rM$ and $p=c/n$.
We then picked each member of $\Phi_V$ to have independent $\mathbb{C}\mathcal{N}(0,\frac{1}{M})$ entries, and we constructed $\Phi_E:=\bigcup_{(i,j)\in E}\{\varphi_i+\zeta^{k}\varphi_j\}_{k=0}^2$.
At this point, we drew a vector $x$ with independent $\mathbb{C}\mathcal{N}(0,\frac{1}{M})$ entries, and performed Phase Retrieval Procedure~A to find an estimate $\tilde{x}$ of $x$.
To be clear, in implementing the first step of Phase Retrieval Procedure~A, we set a tolerance for deletion:
Given $\{|\langle x,\varphi\rangle|^2\}_{\varphi\in\Phi}$, we actually deleted vertices $i\in V$ with $|\langle x,\varphi_i\rangle|^2<10^{-4}$.
After running the remainder of Phase Retrieval Procedure~A, we considered the estimate to be ``good'' if the relative error was small:
\begin{equation*}
\min_{\theta\in[0,2\pi)}\frac{\|\tilde{x}-\mathrm{e}^{\mathrm{i}\theta}x\|}{\|x\|}
<10^{-5}.
\end{equation*}
For each $(r,d)$ pair, we performed 30 trials of this process, and we populated the corresponding cell in Figure~\ref{figure.erdos_renyi} with a shade of gray according to the proportion of good estimates (black indicates that none of the estimates were good, while white indicates that all of the estimates were good).

\begin{figure}[t]
\centering
\includegraphics[width=1\textwidth]{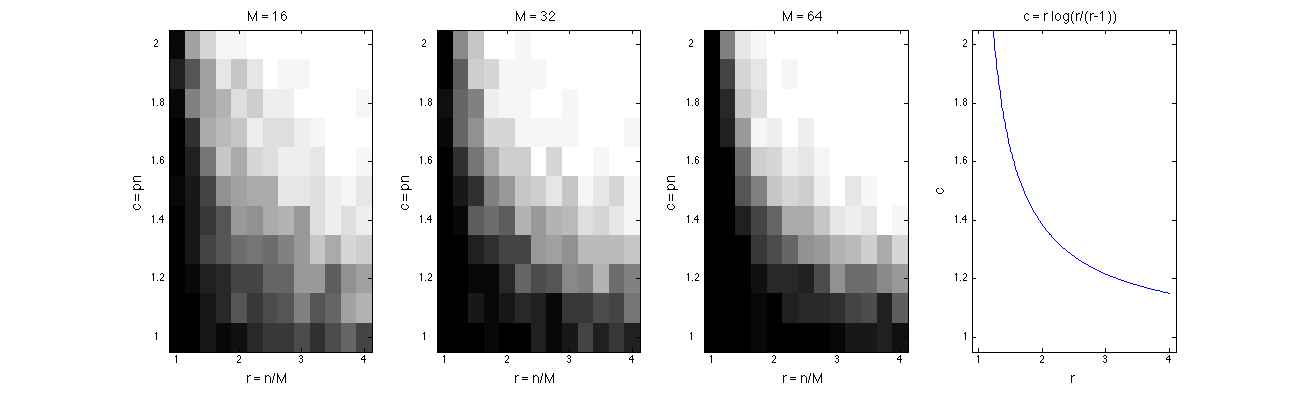}
\caption{
\label{figure.erdos_renyi}
The first three plots indicate the proportion of successful estimates of a random vector using Phase Retrieval Procedure~A; the spatial dimensions for these plots are $M=16$, $32$ and $64$, respectively.
Instead of using a random $d$-regular graph (as prescribed in Measurement Design~A), the graphs for this simulation were Erd\H{o}s-R\'{e}nyi random graphs with $n=rM$ vertices and edge probability $p=c/n$.
The horizontal axes take $r=1:0.25:4$, while the vertical axes take $c=1:0.5:2$.
Darker cells indicate a smaller proportion of successful estimates for that particular choice of $(r,c)$.
The plots appear to be approaching a phase transition, which is illustrated in the fourth plot; here, the curve is given by $c=r\log(r/(r-1))$, and being above or below this curve dictates whether the  Erd\H{o}s-R\'{e}nyi random graph has a component of size at least $M$ with high probability.
Having this large component in the graph essentially determines whether the random vector will be successfully reconstructed.
}
\end{figure}

In observing the results of these simulations in Figure~\ref{figure.erdos_renyi}, we see that the proportion of good estimates appears to go to either $0$ or $1$ depending on whether the pair $(r,d)$ lies above or below the curve given in the fourth plot.
To explain this apparent phase transition, we first note that in the course of running Phase Retrieval Procedure~A for the 12870 total trials, we only had $4$ instances in which a vertex $i$ was deleted due to $|\langle x,\varphi_i\rangle|^2<10^{-4}$.
This is not terribly surprising, considering $x$ and $\varphi_i$ were both drawn to have independent $\mathbb{C}\mathcal{N}(0,\frac{1}{M})$ entries, and therefore do not tend to be almost exactly orthogonal.
What this suggests is that the dominant feature which determines whether reconstruction will be successful is whether the Erd\H{o}s-R\'{e}nyi random graph has a connected component of size at least $M$ (the remainder of Phase Retrieval Procedure~A accumulates only round-off errors, which are negligible here).
At this point, we appeal to a result which appears in Section~5.2 of~\cite{JansonAR:00}, namely that if $c>1$ is held constant, then with high probability, the Erd\H{o}s-R\'{e}nyi random graph with $n$ vertices and edge probability $p=c/n$ has a unique giant connected component of size $(\beta+o(1))n$, where $\beta$ is the unique solution to the equation $\beta+\mathrm{e}^{-\beta c}=1$.
Since in our case, we succeed if the giant component to has size $(\beta+o(1))n\geq M=\frac{n}{r}$, and we fail otherwise, the phase transition should correspond to $\beta=\frac{1}{r}$.
Rearranging $\beta$'s defining equation then gives the formula for the curve in Figure~\ref{figure.erdos_renyi}, namely $c=r\log(r/(r-1))$.

Having established this phase transition, we can use it to minimize the number of measurements.
In particular, the total number of edges in the graph tends to be around $\frac{nc}{2}$, and so the total number of measurements is $N\approx(1+\frac{3}{2}c)n$.
As before, we seek to minimize redundancy:
\begin{equation*}
\frac{N}{M}
=\frac{rN}{n}
\approx r\bigg(1+\frac{3}{2}c\bigg)
=r\bigg(1+\frac{3}{2}r\log\frac{r}{r-1}\bigg).
\end{equation*}
The right-hand side above is minimized when $r\approx 1.28$, in which case we get a redundancy of $\frac{N}{M}\approx 5.02\ll 236$.
The reason for this disparity is simple:
In the previous expander-graph-based analysis, we were chiefly concerned with ensuring that our measurement vectors $\Phi$ lend injective intensity measurements, so that we could reconstruct \emph{any} given signal.
On the other hand, the above analysis demonstrates that we can get away with far fewer measurement vectors if we only need to be able to reconstruct \emph{almost every} signal.
In this sense, these numerical simulations fail to capture the most challenging feature of measurement design for phase retrieval: injectivity (versus unique representation of almost every signal).

\subsection{The noisy case}

In this subsection, we use simulations to demonstrate the stability of Phase Retrieval Procedure~B.
Here, we perform two types of simulations.
In the first simulation, we consider a slightly different noise model so as to compare with different ``phase oracles.''
To be clear, rather than considering noisy intensity measurements of the form $|\langle x,\varphi_\ell\rangle|^2+\nu_\ell$, we instead add the noise \emph{before} the modulus squared: $|\langle x,\varphi_\ell\rangle+\nu_\ell|^2$.
With this change of noise model, we can compare the performance of Phase Retrieval Procedure~B to (i) least-squares estimation from $\Phi^*_Vx+\nu_V$ and (ii) least-squares estimation from $\Phi^*x+\nu$.
This is particularly interesting because the bulk of Phase Retrieval Procedure~B is dedicated to reconstructing (most of) the phases in $\Phi^*_Vx+\nu_V$ which are lost in the modulus squared, and so the comparison with (i) will help us evaluate this portion of the procedure.
Also, (ii) can be viewed as a theoretical limit of sorts, and so this comparison will indicate how much stability we lose with the phases.

For the simulations, we again used the Erd\H{o}s-R\'{e}nyi random graph model instead of a random $d$-regular graph.
Specifically, since $8$-regular random graphs produced measurements with minimum redundancy for Measurement Design~A, and since the corresponding ratio was $r:=\frac{n}{M}\approx 3$, we decided to take $r=3$ and $c=8$ for these simulations.
For each $M=8:4:128$, we picked each of the $n=3M$ members of $\Phi_V$ to have independent $\mathbb{C}\mathcal{N}(0,\frac{1}{M})$ entries, and we constructed $\Phi_E:=\bigcup_{(i,j)\in E}\{\varphi_i+\zeta^{k}\varphi_j\}_{k=0}^2$ in accordance with an Erd\H{o}s-R\'{e}nyi random graph of edge probability $p=\frac{c}{n}=\frac{8}{3M}$.
Next, we picked the signal $x$ to have independent $\mathbb{C}\mathcal{N}(0,\frac{1}{M})$ entries and the noise vector $\nu$ to have independent $\mathbb{C}\mathcal{N}(0,\frac{\sigma^2}{M})$ entries with $\sigma=0.4$.
We then applied Phase Retrieval Procedure~B to $\{|\langle x,\varphi_\ell\rangle+\nu_\ell|^2\}_{\ell=1}^N$ (again, instead of $\{|\langle x,\varphi_\ell\rangle|^2+\nu_\ell\}_{\ell=1}^N$).
Specifically, we took $\alpha=0.9925$ and $\tau=0.1$ as our pruning parameters.
We also performed least-squares estimation from $\Phi^*_Vx+\nu_V$ and from $\Phi^*x+\nu$.
We recorded the relative error and the computation time for each of these three estimates; see Figure~\ref{figure.phase_oracles}.

\begin{figure}[t]
\centering
\includegraphics[width=1\textwidth]{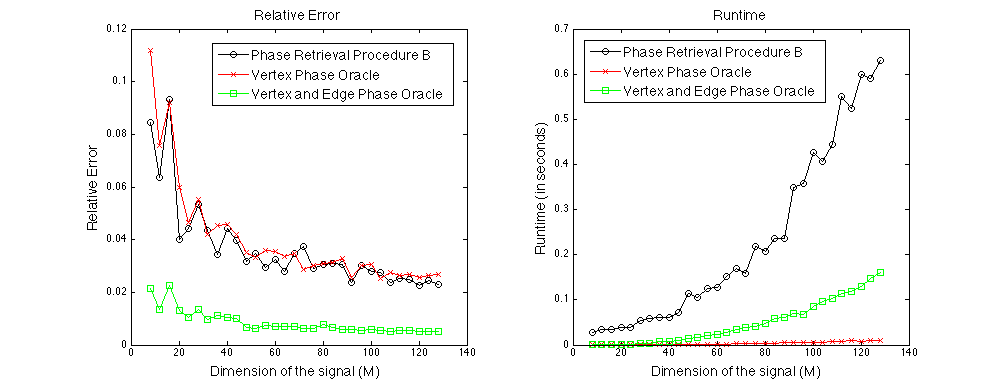}
\caption{
\label{figure.phase_oracles}
Relative error and runtime (in seconds) for three different estimates.
In both plots, the horizontal axis takes $M=8:4:128$.
The black curve denotes Phase Retrieval Procedure~B applied to $\{|\langle x,\varphi_\ell\rangle+\nu_\ell|^2\}_{\ell=1}^N$, where $\nu$ has independent complex-Gaussian entries.
The noise here is added before the modulus squared so that we can compare with two ``phase oracles.''
Specifically, the red curve denotes least-squares estimation from $\Phi^*_Vx+\nu_V$, and the green curve denotes least-squares estimation from $\Phi^*x+\nu$.
Overall, Phase Retrieval Procedure~B does a good job of reconstructing (most of) the vertex phases (since the black and red curves are similar in the relative error plot), but this comes with a modest price in runtime.
}
\end{figure}

Interestingly, Phase Retrieval Procedure~B produces relative errors similar to those gotten by doing least-squares estimation from $\Phi^*_Vx+\nu_V$.
This suggests that the portion of Phase Retrieval Procedure~B which reconstructs (most of) the phases in $\Phi^*_Vx+\nu_V$ performs rather well.
In fact, this shows that the polarization trick of using edges to estimate vertex phases is particularly successful.
Notice that performing least-squares estimation from $\Phi^*x+\nu$ produces a much smaller relative error, as expected.
As far as runtime is concerned, Phase Retrieval Procedure~B is slower than the phase oracles because it takes some time to estimate vertex phases with angular synchronization; however, this is not a substantial difference in runtime, as our procedure still produces an estimate in less than one second.

We would like to point out that in pruning for connectivity, instead of directly applying Algorithm~\ref{alg:find_subexpander}, it sufficed to find the largest surviving component, since in our trials, this component always had spectral gap larger than $\tau=0.1$.
We suspect that this is an artifact of the random graph, as this will certainly not happen in general.

The second simulation we ran to demonstrate stability was a comparison with the standard alternating projections routine:
After assigning arbitrary (say, positive) phases to the intensity measurements, iteratively project onto the column space of $\Phi^*$ and then project onto the nonconvex set of vectors $y$ whose entry magnitudes have the given intensity measurements.
We applied both algorithms to $\{|\langle x,\varphi_\ell\rangle|^2+\nu_\ell\}_{\ell=1}^N$, where the $\nu_\ell$'s are independent $\mathcal{N}(0,\frac{\sigma^2}{M})$ entries with $\sigma=0.4$.
Other than the different phase retrieval algorithm and noise model, our simulation protocol was identical to the previous one, and the results are presented in Figure~\ref{figure.alt_proj}.

\begin{figure}[t]
\centering
\includegraphics[width=1\textwidth]{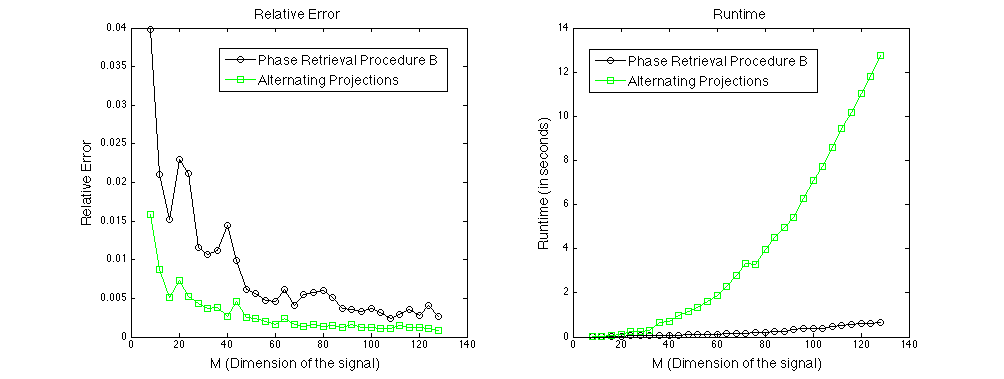}
\caption{
\label{figure.alt_proj}
Relative error and runtime (in seconds) for two different estimates.
In both plots, the horizontal axis takes $M=8:4:128$.
The black curve denotes Phase Retrieval Procedure~B applied to $\{|\langle x,\varphi_\ell\rangle|^2+\nu_\ell\}_{\ell=1}^N$, where $\nu$ has independent Gaussian entries.
The green curve denotes alternating projections applied to the same noisy intensity measurements.
Overall, alternating projections does surprisingly well at doing phase retrieval in these instances, though requiring a bit more runtime.
Note that the relative error of Phase Retrieval Procedure~B is different here than in Figure~\ref{figure.phase_oracles} because in this case, the noise is added after the modulus squared.
}
\end{figure}

The most striking thing about this simulation is that alternating projections consistently produces a slightly better estimate than Phase Retrieval Procedure~B.
For comparison, we also ran alternating projections using only the $3M$ vertex measurement vectors $\Phi_V$, and the relative errors were consistently on the order of $1$, i.e., alternating projections consistently stalled in this case.
This suggests that there is some fundamental quality about the polarized measurement vectors $\Phi$ which makes them particularly well-suited for alternating projections, and we intend to study this in the future.
Regardless, alternating projections took a lot longer to terminate; to be clear, we terminated the loop once applying both projections moved the estimate by less than $10^{-3}$, or by the $100$th iteration, whichever occurred first.

\section{Concluding remarks}

This paper provides a new way to perform phase retrieval, and our main result (Theorem~\ref{thm.main result}) shows that our method is stable.
In comparing with the stability result of~\cite{CandesSV:11}, we note that neither result is completely satisfying when viewed from the perspective of application:
In the real world, you are given a noise level and an acceptable level of estimate error, and you are asked to meet these specifications with signal processing techniques.
For phase retrieval, the available guarantees fail to prescribe a measurement design that overcomes a given noise level---rather, they merely establish that with sufficiently many measurements, there exists some level of stability, i.e., $K$ in Theorem~\ref{thm.main result} or $C_0$ in Theorem~1.2 of~\cite{CandesSV:11}.
This reveals a gap in what is known about stability in phase retrieval, and we leave this for future work. 

Admittedly, there are several gaps remaining between modern theory and application of phase retrieval.
For example, thoughout this paper, it is assumed that the user has complete knowledge of the measurement design, but this is not always possible in practice.
This can be resolved in part with new stability results which account for ``noise'' in the measurement design (this is sometimes called mismatch error).

One might feel that our phase retrieval algorithms are slightly unsatisfying because we perform hard thresholds to remove vertices according to how small or large the corresponding measurements are.
Alternatively, there could very well be a way to more smoothly weight these measurements according to our confidence in them, and such weightings are already accounted for in the theory of angular synchronization~\cite{BandeiraSS:12}.
However, we decided to use hard thresholds because they greatly simplify the analysis of projective uniformity (though the analysis is still rather technical).

While the worst-case analysis we provide here is useful in many applications (and enables a comparison with the worst-case stability results of~\cite{CandesSV:11}), stochastic noise is a more appropriate model in other applications.
We believe that the phase retrieval procedure of this paper will perform substantially better in the average case, but we leave this analysis for future work.
Also, a notable distinction between our measurement designs in the noiseless and noisy cases is the presence of a log factor in the number of measurements used.
However, we believe this factor is an artifact of our current analysis, and we intend to remove it in the future.

\section{Appendix}

\subsection{Graph pruning}

This section proves the following guarantee:

\begin{theorem}
\label{thm:GraphTrimming_MAIN}
Take proportions $p\geq q\geq\frac{2}{3}$, and consider a regular graph $G=(V,E)$ with spectral gap $\lambda_2>g(p,q):=1-2(q(1-q)-(1-p))$.
After Algorithm~\ref{alg:takingedgestotakevert} removes at most $(1-p)|V|$ vertices from $G$, then setting $\tau=\frac{1}{8}(\lambda_2-g(p,q))^2$, Algorithm~\ref{alg:find_subexpander} outputs a subgraph with at least $q|V|$ vertices.
\end{theorem}

To prove this theorem, we will apply a graph version of the Cheeger inequality, which provides a guarantee for Algorithm~\ref{alg:find_spectralclustering}:

\begin{theorem}[Constructive Cheeger inequality~\cite{Chung:10}]\label{theorem:CI_ClassicalL0Const}
Consider a graph $G=(V,E)$ with spectral gap $\lambda_2$.
Then Algorithm~\ref{alg:find_spectralclustering} outputs a set of vertices $S$ such that $h(S)\leq\sqrt{2\lambda_2}$.
\end{theorem}

\begin{proof}[Proof of Theorem~\ref{thm:GraphTrimming_MAIN}]
First, Algorithm~\ref{alg:takingedgestotakevert} removes a set of vertices, which we denote by $S_0$.
In applying Algorithm~\ref{alg:find_subexpander}, the $i$th step of the while loop removes another set of vertices $S_i$.
We claim this while loop will end with $|\bigcup_{i\geq0}S_i|<(1-q)|V|$.
Supposing to the contrary, consider the first $k$ for which $S:=\bigcup_{i=0}^kS_i$ has at least $(1-q)|V|$ vertices.
Then since each iteration of the while loop removes at most half of the remaining vertices, we have $(1-q)|V|\leq|S|\leq(1-\frac{q}{2})|V|$.

To derive a contradiction, we will find incompatible upper and lower bounds on $E(S,S^\mathrm{c})$.
For the upper bound, we apply the fact that $G$ is $d$-regular along with the definition of $h(S_i)$ to get
\begin{equation*}
E(S,S^\mathrm{c})
=E(S_0,S^\mathrm{c})+\sum_{i=1}^kE(S_i,S^\mathrm{c})
\leq \mathrm{vol}(S_0)+\sum_{i=1}^kE(S_i,S_i^\mathrm{c})
\leq d|S_0|+\sum_{i=1}^kh(S_i)\vol(S_i).
\end{equation*}
Next, Theorem~\ref{theorem:CI_ClassicalL0Const} bounds each $h(S_i)$ in terms of the spectral gap the remaining graph, which is necessarily less than $\tau$ by the condition of the while loop.
Thus, we continue:
\begin{equation}
\label{eq.pruning upper bound}
E(S,S^\mathrm{c})
< d|S_0|+\sum_{i=1}^k \sqrt{2\tau} d|S_i|
\leq d(1-p)|V|+\sqrt{2\tau} d|S|
\leq d|V|\Big((1-p)+\sqrt{2\tau}(1-\tfrac{q}{2})\Big).
\end{equation}
For the lower bound, we use the expander mixing lemma, which says 
\begin{equation*}
\Big|E(S,S^\mathrm{c})-\tfrac{d}{|V|}|S||S^\mathrm{c}|\Big|\leq d(1-\lambda_2)\sqrt{|S||S^\mathrm{c}|}.
\end{equation*}
Since the function $x\mapsto x(1-x)$ is concave down, we know the minimum subject to $1-q\leq x\leq 1-\frac{q}{2}$ is achieved at an endpoint of this interval.
Thus, evaluating the function at $|S|/|V|$ gives $|S||S^\mathrm{c}|/|V|^2\geq\min\{q(1-q),\frac{q}{2}(1-\frac{q}{2})\}=q(1-q)$, where the last step follows from the fact that $q\geq\frac{2}{3}$.
Applying this and $|S||S^\mathrm{c}|\leq|V|^2/4$ to the expander mixing lemma then gives
\begin{equation}
\label{eq.pruning lower bound}
E(S,S^\mathrm{c})
\geq \tfrac{d}{|V|}|S||S^\mathrm{c}|-d(1-\lambda_2)\sqrt{|S||S^\mathrm{c}|}
\geq d|V|q(1-q)-d(1-\lambda_2)\tfrac{|V|}{2}
= d|V|\Big(q(1-q)-\tfrac{1-\lambda_2}{2}\Big).
\end{equation}
Finally, we combine \eqref{eq.pruning upper bound} and \eqref{eq.pruning lower bound} to get
\begin{equation*}
q(1-q)-\tfrac{1-\lambda_2}{2}
<(1-p)+\sqrt{2\tau}(1-\tfrac{q}{2}),
\end{equation*}
which, as substitution reveals, contradicts our choice for $\tau$.
\qquad
\end{proof}

\subsection{Angular synchronization}

This section proves the following guarantee:

\begin{theorem}
\label{thm:ang_synch_appendix}
Consider a graph $G=(V,E)$ with spectral gap $\tau>0$, and define $\|\theta\|_\mathbb{T}:=\min_{k\in\mathbb{Z}}|\theta-2\pi k|$ for all angles $\theta\in\mathbb{R}/2\pi\mathbb{Z}$.
Given the weighted adjacency matrix $A_1$~\eqref{eq.weighted adjacency}, then Algorithm~\ref{alg:ang_synch_appendix} outputs $u\in\mathbb{C}^{|V|}$ with unit-modulus entries such that, for some phase $\theta\in\mathbb{R}/2\pi\mathbb{Z}$,
\begin{equation*}
\sum_{i\in V}\big\|\arg(u_i)-\arg(\langle x,\varphi_i\rangle)-\theta\big\|_\mathbb{T}^2
\leq\frac{C\|\varepsilon\|^2}{\tau^2 P^2},
\end{equation*}
where $P:=\min_{\{i,j\}\in E}|\overline{\langle x,\varphi_i\rangle}\langle x,\varphi_j\rangle+\varepsilon_{ij}|$ and $C$ is a universal constant.
\end{theorem}

In terms of phases, we are interested in reconstructing $\omega^*:V\to \mathbb{R}/2\pi\mathbb{Z}$ such that $\omega^*_i = \arg(\langle x,\varphi_i\rangle)$ for every $i\in V$. 
For some graph $G=(V,E)$, we are given a weighted adjacency matrix $A_1$ \eqref{eq.weighted adjacency}, which encodes edge measurements:
\begin{equation*}
\rho_{ij}
:=\omega^*_j-\omega^*_i+e_{ij}
\end{equation*}
for every $(i,j)\in E$; here, $e_{ij}\in\mathbb{R}/2\pi\mathbb{Z}$ is angular noise, and we note that $e_{ji}=-e_{ij}$. 
We measure the size of $e=\{e_{ij}\}_{\{i,j\}\in E}$ in terms of its components: $\|e\|_\mathbb{T}^2:=\sum_{\{i,j\}\in E}\|e_{ij}\|_\mathbb{T}^2$.

We will prove Theorem \ref{thm:ang_synch_appendix} using a recent Cheeger inequality for the connection Laplacian~\cite{BandeiraSS:12}.
In effect, this result provides a guarantee for Algorithm \ref{alg:ang_synch_appendix} in terms of a certain objective function.
To be precise, given $\omega:V\to \mathbb{R}/2\pi\mathbb{Z}$, we define
\begin{equation*}
\eta(\omega)
=\sum_{\{i,j\}\in E}|\mathrm{e}^{\mathrm{i}((\omega_j-\omega_i)-\rho_{ij})}-1|^2.
\end{equation*}
Note that one way to reconstruct $\omega^*$ is to minimize this quantity.
Indeed, $\eta(\omega)$ is small when the angular differences $\omega_j-\omega_i$ are close to the measured differences $\rho_{ij}$, and in the noiseless case, $\eta(\omega)=0$ precisely when $\omega=\omega^*$, provided the graph $G$ is connected.
In terms of this objective function, the following guarantee ensures that the output of Algorithm~\ref{alg:ang_synch_appendix} is no worse than a constant multiple of optimal:

\begin{theorem}[Cheeger inequality for the connection Laplacian~\cite{BandeiraSS:12}]
\label{CheegerInequalityConLap}
Consider a graph $G=(V,E)$ with spectral gap $\tau>0$. 
Given the weighted adjacency matrix $A_1$~\eqref{eq.weighted adjacency}, which encodes edge measurements $\rho:E\to\mathbb{R}/2\pi\mathbb{Z}$, then Algorithm~\ref{alg:ang_synch_appendix} outputs $u\in\mathbb{C}^{|V|}$, which encodes $\hat\omega:V\rightarrow\mathbb{R}/2\pi\mathbb{Z}$ such that
\begin{equation}
\eta\big(\hat\omega\big)
\leq \frac{C'}{\tau} \min_{\omega:V\to\TT}\eta(\omega),
\end{equation}
where $C'$ is a universal constant.
\end{theorem}

We start with two lemmas:

\begin{lemma}
\label{lemma.bandeira inequality}
For every $a,b\in\mathbb{R}/2\pi \mathbb{Z}$, we have $\frac{1}{2}\|a\|_\mathbb{T}^2-\|b\|_\mathbb{T}^2\leq\|a-b\|_\mathbb{T}^2$.
\end{lemma}

\begin{proof}
By abuse of notation, we identify $a$ and $b$ with their coset representatives in $[-\pi,\pi)$.
Since
\begin{equation*}
0
\leq\tfrac{1}{2}(a-2b)^2
=(a-b)^2-(\tfrac{1}{2}a^2-b^2),
\end{equation*}
then rearranging gives the following inequality:
\begin{equation*}
\tfrac{1}{2}\|a\|_\mathbb{T}^2-\|b\|_\mathbb{T}^2
=\tfrac{1}{2}a^2-b^2
\leq(a-b)^2
=\|a-b\|_\mathbb{T}^2,
\end{equation*}
where the last equality requires $a-b\in[-\pi,\pi]$.
Otherwise, we note that
\begin{equation*}
0
\leq\tfrac{1}{2}(a-2b\pm2\pi)^2+2\pi(\pi\pm a)
=(a-b\pm2\pi)^2-(\tfrac{1}{2}a^2-b^2),
\end{equation*}
and so rearranging gives
\begin{equation*}
\tfrac{1}{2}\|a\|_\mathbb{T}^2-\|b\|_\mathbb{T}^2
=\tfrac{1}{2}a^2-b^2
\leq(a-b\pm2\pi)^2,
\end{equation*}
one of which equals $\|a-b\|_\mathbb{T}^2$ depending on whether $a-b<-\pi$ or $a-b>\pi$.
\qquad
\end{proof}

\begin{lemma}
\label{lemma.unit circle lemma}
For every $a,b\in\mathbb{C}$ with $|a|=1$, we have $|a-\frac{b}{|b|}|\leq 2|a-b|$.
\end{lemma}

\begin{proof}
The reverse triangle inequality gives
$|\tfrac{b}{|b|}-b|
=|1-|b||
=||a|-|b||
\leq |a-b|$,
and so the triangle inequality gives
$|a-\tfrac{b}{|b|}|
\leq |a-b|+|b-\tfrac{b}{|b|}|
\leq 2|a-b|$.
\qquad
\end{proof}

\begin{proof}[Proof of Theorem~\ref{thm:ang_synch_appendix}]
Identifying $\theta\in\mathbb{R}/2\pi\mathbb{Z}$ with its coset representative in $[-\pi,\pi)$, we have $\|\theta\|_\mathbb{T}=|\theta|$, and a double-angle formula gives $|\mathrm{e}^{\mathrm{i}\theta}-1|=2\sin\frac{|\theta|}{2}$.
From these, it follows that
\begin{equation}
\label{eq.norm equivalence t i}
\tfrac{2}{\pi}\|\theta\|_\mathbb{T}
\leq|\mathrm{e}^{\mathrm{i}\theta}-1|
\leq\|\theta\|_\mathbb{T}.
\end{equation}
This relationship will allow us to apply Theorem~\ref{CheegerInequalityConLap}.
To this end, for notational convenience, we define $\gamma_i:=\arg(u_i)-\arg(\langle x,\varphi_i\rangle)$ for every $i\in V$.
The right-hand inequality of \eqref{eq.norm equivalence t i} and Lemma~\ref{lemma.bandeira inequality} together give 
\begin{equation}
\label{eq.sych 1}
\tfrac{1}{2}\!\!\sum_{\{i,j\}\in E}\!\!|\mathrm{e}^{\mathrm{i}(\gamma_j-\gamma_i)}-1|^2-\|e\|_\mathbb{T}^2
\leq\sum_{\{i,j\}\in E}\Big(\tfrac{1}{2}\|\gamma_j-\gamma_i\|_\mathbb{T}^2-\|e_{ij}\|_\mathbb{T}^2\Big)
\leq\sum_{\{i,j\}\in E}\|\gamma_j-\gamma_i-e_{ij}\|_\mathbb{T}^2.
\end{equation}
Denoting $\hat\omega_i:=\arg(u_i)$ and $\omega^*_i:=\arg(\langle x,\varphi_i\rangle)$, then the definition of $\eta$ gives
\begin{equation*}
\eta(\hat\omega)
=\sum_{\{i,j\}\in E}|\mathrm{e}^{\mathrm{i}(\gamma_j-\gamma_i-e_{ij})}-1|^2.
\end{equation*}
With this, we continue \eqref{eq.sych 1} by applying the left-hand inequality of \eqref{eq.norm equivalence t i}:
\begin{equation}
\label{eq.sych 2}
\tfrac{1}{2}\!\!\sum_{\{i,j\}\in E}\!\!|\mathrm{e}^{\mathrm{i}(\gamma_j-\gamma_i)}-1|^2-\|e\|_\mathbb{T}^2
\leq\tfrac{\pi^2}{4}\eta(\hat\omega)
\leq\tfrac{\pi^2}{4}\tfrac{C'}{\tau}\min_{\omega:V\rightarrow\mathbb{T}}\eta(\omega)
\leq\tfrac{\pi^2}{4}\tfrac{C'}{\tau}\eta(\omega^*),
\end{equation}
where the second inequality follows from Theorem~\ref{CheegerInequalityConLap}.
Furthermore, the right-hand inequality of \eqref{eq.norm equivalence t i} gives
\begin{equation*}
\eta(\omega^*)
=\sum_{\{i,j\}\in E}|\mathrm{e}^{-\mathrm{i}e_{ij}}-1|^2
\leq\|e\|_\mathbb{T}^2,
\end{equation*}
and so combining this with \eqref{eq.sych 2} gives
\begin{equation}
\label{eq.sych 3}
\sum_{\{i,j\}\in E}|\mathrm{e}^{\mathrm{i}(\gamma_j-\gamma_i)}-1|^2
\leq2(\tfrac{\pi^2}{4}\tfrac{C'}{\tau}+1)\|e\|_\mathbb{T}^2.
\end{equation}
At this point, take $\alpha:=\frac{1}{\vol(G)}\sum_{i\in V}\deg(i)\mathrm{e}^{\mathrm{i}\gamma_i}$, define $v$ entrywise by $v_i:=\mathrm{e}^{\mathrm{i}\gamma_i}$, and set $w=v-\alpha1$.
Then 
\begin{equation*}
1^*Dw
=\sum_{i\in V}\deg(i)(\mathrm{e}^{\mathrm{i}\gamma_i}-\alpha)
=0, 
\end{equation*}
i.e., $D^{1/2}w$ is orthogonal to $D^{1/2}1$.
Also, since $(D-A)1=0$, we have
\begin{equation*}
L(D^{1/2}1)
=(I-D^{-1/2}AD^{-1/2})D^{1/2}1
=D^{-1/2}(D-A)D^{-1/2}D^{1/2}1
=0,
\end{equation*}
and so $D^{1/2}w$ is orthogonal to a first eigenvector of $L$, considering $L$ is positive semidefinite.
Thus,
\begin{equation*}
\frac{(D^{1/2}w)^*L(D^{1/2}w)}{w^*Dw}
\geq\min_{\substack{y\in\mathbb{C}^{|V|}\\y\perp D^{1/2}1}}\frac{y^*Ly}{y^*y}
=\tau.
\end{equation*}
Rearranging then gives
\begin{equation*}
\tau w^*Dw
\leq (D^{1/2}w)^*L(D^{1/2}w)
=(v-\alpha1)^*(D-A)(v-\alpha1)
=v^*(D-A)v.
\end{equation*}
Continuing, we apply the definitions of $D$ and $A$ to get
\begin{equation*}
\tau w^*Dw
\leq 2|E|-\sum_{i\in V}\sum_{j\in V}\mathrm{e}^{-\mathrm{i}\gamma_j}A[i,j]\mathrm{e}^{\mathrm{i}\gamma_i}
=\sum_{\{i,j\}\in E}\Big(2-\mathrm{e}^{\mathrm{i}(\gamma_i-\gamma_j)}-\mathrm{e}^{\mathrm{i}(\gamma_j-\gamma_i)}\Big)
=\sum_{\{i,j\}\in E}|\mathrm{e}^{\mathrm{i}\gamma_j}-\mathrm{e}^{\mathrm{i}\gamma_i}|^2.
\end{equation*}
Next, we factor $\mathrm{e}^{\mathrm{i}\gamma_i}$ from the inside of each term and apply \eqref{eq.sych 3} to get
\begin{equation}
\label{eq.sych 4}
\tau\sum_{i\in V}\deg(i)|\mathrm{e}^{\mathrm{i}\gamma_i}-\alpha|^2
=\tau w^*Dw
\leq \sum_{\{i,j\}\in E}|\mathrm{e}^{\mathrm{i}(\gamma_j-\gamma_i)}-1|^2
\leq 2(\tfrac{\pi^2}{4}\tfrac{C'}{\tau}+1)\|e\|_\mathbb{T}^2.
\end{equation}
From here, we proceed in two cases.
First, when $\alpha\neq0$, we may take $\theta:=\arg(\alpha)$.
Then the left-hand inequality of \eqref{eq.norm equivalence t i} and Lemma~\ref{lemma.unit circle lemma} give
\begin{equation}
\label{eq.sych 5}
\sum_{i\in V}\|\gamma_i-\theta\|_\mathbb{T}^2
\leq\tfrac{\pi^2}{4}\sum_{i\in V}|\mathrm{e}^{\mathrm{i}(\gamma_i-\theta)}-1|^2
=\tfrac{\pi^2}{4}\sum_{i\in V}|\mathrm{e}^{\mathrm{i}\gamma_i}-\tfrac{\alpha}{|\alpha|}|^2
\leq \pi^2\sum_{i\in V}|\mathrm{e}^{\mathrm{i}\gamma_i}-\alpha|^2
\leq \pi^2\sum_{i\in V}\deg(i)|\mathrm{e}^{\mathrm{i}\gamma_i}-\alpha|^2,
\end{equation}
where the last inequality uses the fact that $\deg(i)\geq1$ for every $i\in V$, i.e., the graph $G$ has no isolated vertex since $G$ is connected, which follows from the fact that $\tau>0$.
In the case where $\alpha=0$, we may arbitrarily take $\theta=0$.
Then similar analysis yields
\begin{equation*}
\sum_{i\in V}\|\gamma_i-\theta\|_\mathbb{T}^2
\leq\tfrac{\pi^2}{4}\sum_{i\in V}|\mathrm{e}^{\mathrm{i}\gamma_i}-1|^2
\leq\pi^2|V|
\leq\pi^2\sum_{i\in V}\deg(i)
=\pi^2\sum_{i\in V}\deg(i)|\mathrm{e}^{\mathrm{i}\gamma_i}-\alpha|^2,
\end{equation*}
where in this case, the second inequality applies the triangle inequality to each term instead of Lemma~\ref{lemma.unit circle lemma}.
Note that both cases produce the same bound on $\sum_{i\in V}\|\gamma_i-\theta\|_\mathbb{T}^2$, which we now conclude by combining \eqref{eq.sych 4} and \eqref{eq.sych 5}:
\begin{equation*}
\sum_{i\in V}\|\gamma_i-\theta\|_\mathbb{T}^2
\leq\tfrac{2\pi^2}{\tau}(\tfrac{\pi^2}{4}\tfrac{C'}{\tau}+1)\|e\|_\mathbb{T}^2
\leq \tfrac{2\pi^2}{\tau}(\tfrac{\pi^2}{4}\tfrac{C'}{\tau}+1) \cdot \tfrac{\pi^2\|\varepsilon\|^2}{P^2}.
\end{equation*}
The last inequality follows from Lemma~\ref{proposition:additiveVsangularNOISE}, taking $z=\overline{\langle x,\varphi_i\rangle}\langle x,\varphi_j\rangle+\varepsilon_{ij}$ and $\varepsilon=-\varepsilon_{ij}$.
\qquad
\end{proof}

\subsection{Projective uniformity}

This section is motivated by Theorem~\ref{thm:ang_synch_appendix} of the previous section, which exhibits significant dependence on the size of $P$.
Here, we show how Algorithm~\ref{alg:takingedgestotakevert} ensures that $P$ will not too small, and our guarantee will be in terms of the following noise-robust version of projective uniformity:

\begin{definition}[Projective uniformity with noise]\label{definition:NOISYPhaselessRadialSpreadConst}
Consider a graph $G=(V,E)$ and $M\times |V|$ matrix $\Phi$, and for some proportion $\alpha\in(0,1)$, signal $x\in\mathbb{C}^M$ and noise $\varepsilon=\{\varepsilon_{ij}\}_{\{i,j\}\in E}$, let $\mathcal{J}(\alpha,x,\varepsilon)$ denote the set of vertices that remain after applying Algorithm~\ref{alg:takingedgestotakevert}.
Then the \textit{projective uniformity with noise} of $\Phi$ is
\begin{equation*}
\mathrm{PUN}(\Phi;\alpha,\varepsilon)
=\min_{\substack{x\in\mathbb{C}^M\\\|x\|=1}}\min_{\{i,j\}\in\mathcal{J}(\alpha,x,\varepsilon)}|\overline{\langle x,\varphi_i\rangle}\langle x,\varphi_j\rangle+\varepsilon_{ij}|.
\end{equation*}
\end{definition}

\begin{theorem}
\label{th:PUN_bigTheorem}
Consider a graph $G=(V,E)$ with $|V|\sim CM\log M$ for some constant $C$, and draw the entries of an $M\times|V|$ matrix $\Phi$ independently from $\mathbb{C}\mathcal{N}(0,\frac{1}{M})$.
Then for each proportion $\alpha<1-\frac{1}{2C}$, there exists a constant $C'>0$ such that, with overwhelming probability,
\begin{equation*}
\max\Big\{\|\varepsilon\|_2,\sqrt{M}~\mathrm{PUN}(\Phi;\alpha,\varepsilon)\Big\}
\geq \frac{C'}{\sqrt{M}}
\end{equation*}
for every noise vector $\varepsilon=\{\varepsilon_{ij}\}_{\{i,j\}\in E}$.
\end{theorem}

To prove this theorem, we apply the following lemma, which follows from concentration-of-measure arguments that we provide later:

\begin{lemma}
\label{th:RadialSpreadConstant_bigLemma}
Take $n\sim CM\log M$ for some constant $C$, and draw the entries of an $M\times n$ matrix $\Phi$ independently from $\mathbb{C}\mathcal{N}(0,\frac{1}{M})$.
Then for each proportion $\alpha<1-\frac{1}{2C}$, there exists a constant $C'>0$ such that
\begin{equation*}
\mathrm{PU}(\Phi;\alpha)
\geq \frac{C'}{M}
\end{equation*}
with overwhelming probability.
\end{lemma}

\begin{proof}[Proof of Theorem~\ref{th:PUN_bigTheorem}]
Denote $n:=|V|$.
Note that if 
\begin{equation}
\label{eq.big theorem trivial case}
\|\varepsilon\|_2>\sqrt{\tfrac{1-\alpha}{2}n}\cdot\tfrac{1}{2}\mathrm{PU}(\Phi;\alpha+\tfrac{1+\alpha}{2}), 
\end{equation}
then by Lemma~\ref{th:RadialSpreadConstant_bigLemma}, we have $\|\varepsilon\|_2\geq \frac{C'}{\sqrt{M}}$ with overwhelming probability for some constant $C'>0$.
It remains to consider the case where \eqref{eq.big theorem trivial case} does not hold.
To this end, define
\begin{equation*}
\mathcal{J}_1(\alpha,x,\varepsilon)
:=\argmax_{\substack{\mathcal{J}\subseteq\{1,\ldots,n\}\\|\mathcal{J}|\geq\alpha n}}\min_{\{i,j\}\in\mathcal{J}}|\overline{\langle x,\varphi_i\rangle}\langle x,\varphi_j\rangle+\varepsilon_{ij}|.
\end{equation*}
Recalling the definition of $\mathcal{J}(\alpha,x,\varepsilon)$ in Definition~\ref{definition:NOISYPhaselessRadialSpreadConst}, we claim that $\mathcal{J}(\alpha,x,\varepsilon)\subseteq\mathcal{J}_1(\alpha,x,\varepsilon)$.
To see this, label each edge $\{i,j\}\in E$ with $|\overline{\langle x,\varphi_i\rangle}\langle x,\varphi_j\rangle+\varepsilon_{ij}|$.
Then by definition, $V\setminus\mathcal{J}_1(\alpha,x,\varepsilon)$ neighbors more of the smallest edges than any other collection of $n-\lceil \alpha n\rceil=\lfloor (1-\alpha)n\rfloor$ vertices.
By comparison, Algorithm~\ref{alg:takingedgestotakevert} effectively deletes the smallest edge $\{i,j\}$ by deleting both of its incident vertices $i$ and $j$, one of which must be in $\mathcal{J}_1(\alpha,x,\varepsilon)$.
The smallest edge in the remaining graph is guaranteed to not touch either $i$ or $j$, but rather some $k\in\mathcal{J}_1(\alpha,x,\varepsilon)$, provided $\mathcal{J}_1(\alpha,x,\varepsilon)$ has not yet been completely removed from the graph.
After $\lfloor (1-\alpha)n\rfloor$ iterations, then by the pigeonhole principle, all vertices in $V\setminus\mathcal{J}_1(\alpha,x,\varepsilon)$ will be removed, meaning $\mathcal{J}(\alpha,x,\varepsilon)\subseteq\mathcal{J}_1(\alpha,x,\varepsilon)$, as claimed.
This implies that
\begin{equation*}
\min_{\{i,j\}\in\mathcal{J}(\alpha,x,\varepsilon)}|\overline{\langle x,\varphi_i\rangle}\langle x,\varphi_j\rangle+\varepsilon_{ij}|
\geq \min_{\{i,j\}\in\mathcal{J}_1(\alpha,x,\varepsilon)}|\overline{\langle x,\varphi_i\rangle}\langle x,\varphi_j\rangle+\varepsilon_{ij}|,
\end{equation*}
and so minimizing over all unit vectors $x$ gives
\begin{equation}
\label{eq.big theorem ingredient 1}
\mathrm{PUN}(\Phi;\alpha,\varepsilon)
\geq\min_{\substack{x\in\mathbb{C}^M\\\|x\|=1}}\max_{\substack{\mathcal{J}\subseteq\{1,\ldots,n\}\\|\mathcal{J}|\geq\alpha n}}\min_{\{i,j\}\in\mathcal{J}}|\overline{\langle x,\varphi_i\rangle}\langle x,\varphi_j\rangle+\varepsilon_{ij}|.
\end{equation}
Next, define
\begin{equation*}
\mathcal{J}_2(\alpha,x,\varepsilon)
:=\argmax_{\substack{\mathcal{J}\subseteq\{1,\ldots,n\}\\|\mathcal{J}|\geq(\alpha+\frac{1-\alpha}{2})n}}\min_{\{i,j\}\in\mathcal{J}}|\overline{\langle x,\varphi_i\rangle}\langle x,\varphi_j\rangle|.
\end{equation*}
Since in this case \eqref{eq.big theorem trivial case} does not hold, there are at most $\frac{1-\alpha}{2}n$ edges $\{i,j\}\in E$ with $|\varepsilon_{ij}|\geq\frac{1}{2}\mathrm{PU}(\Phi;\alpha+\tfrac{1+\alpha}{2})$.
Some of these edges are induced by $\mathcal{J}_2(\alpha,x,\varepsilon)$; as such, delete at most $\frac{1-\alpha}{2}n$ vertices from $\mathcal{J}_2(\alpha,x,\varepsilon)$ which are incident to all of these edges, and denote the remaining vertices by $\mathcal{J}_3(\alpha,x,\varepsilon)$.
By construction, we have $|\mathcal{J}_3(\alpha,x,\varepsilon)|\geq\alpha n$, and also the edges $\{i,j\}$ induced by $\mathcal{J}_3(\alpha,x,\varepsilon)$ satisfy $|\varepsilon_{ij}|<\frac{1}{2}\mathrm{PU}(\Phi;\alpha+\tfrac{1+\alpha}{2})$.
These facts combined with the triangle inequality then give
\begin{align}
\max_{\substack{\mathcal{J}\subseteq\{1,\ldots,n\}\\|\mathcal{J}|\geq \alpha n}}\min_{\{i,j\}\in\mathcal{J}}|\overline{\langle x,\varphi_i\rangle}\langle x,\varphi_j\rangle+\varepsilon_{ij}|
\nonumber
&\geq\min_{\{i,j\}\in\mathcal{J}_3(\alpha,x,\varepsilon)}|\overline{\langle x,\varphi_i\rangle}\langle x,\varphi_j\rangle+\varepsilon_{ij}|\\
\nonumber
&\geq\min_{\{i,j\}\in\mathcal{J}_3(\alpha,x,\varepsilon)}\Big(|\overline{\langle x,\varphi_i\rangle}\langle x,\varphi_j\rangle|-|\varepsilon_{ij}|\Big)\\
\label{eq.big theorem ingredient 2}
&\geq\min_{\{i,j\}\in\mathcal{J}_3(\alpha,x,\varepsilon)}|\overline{\langle x,\varphi_i\rangle}\langle x,\varphi_j\rangle|-\tfrac{1}{2}\mathrm{PU}(\Phi;\alpha+\tfrac{1+\alpha}{2}).
\end{align}
Continuing, we use the fact that $\mathcal{J}_3(\alpha,x,\varepsilon)\subseteq\mathcal{J}_2(\alpha,x,\varepsilon)$ to get
\begin{align}
\min_{\{i,j\}\in\mathcal{J}_3(\alpha,x,\varepsilon)}|\overline{\langle x,\varphi_i\rangle}\langle x,\varphi_j\rangle|
\nonumber
&\geq\min_{\{i,j\}\in\mathcal{J}_2(\alpha,x,\varepsilon)}|\overline{\langle x,\varphi_i\rangle}\langle x,\varphi_j\rangle|\\
\nonumber
&\geq\min_{\substack{x\in\mathbb{C}^M\\\|x\|=1}}\max_{\substack{\mathcal{J}\subseteq\{1,\ldots,n\}\\|\mathcal{J}|\geq(\alpha+\frac{1-\alpha}{2}) n}}\min_{\{i,j\}\in\mathcal{J}}|\overline{\langle x,\varphi_i\rangle}\langle x,\varphi_j\rangle|\\
\label{eq.big theorem ingredient 3}
&\geq\mathrm{PU}(\Phi;\alpha+\tfrac{1+\alpha}{2}).
\end{align}
Combining \eqref{eq.big theorem ingredient 1}, \eqref{eq.big theorem ingredient 2} and \eqref{eq.big theorem ingredient 3} then gives $\mathrm{PUN}(\Phi;\alpha,\varepsilon)\geq\frac{1}{2}\mathrm{PU}(\Phi;\alpha+\tfrac{1+\alpha}{2})$, and so the result follows directly from Lemma~\ref{th:RadialSpreadConstant_bigLemma}.
\qquad
\end{proof}

We conclude this section with the rather technical proof of Lemma~\ref{th:RadialSpreadConstant_bigLemma}:

\begin{proof}[Proof of Lemma~\ref{th:RadialSpreadConstant_bigLemma}]
We seek a bound on the probability that $\mathrm{PU}(\Phi;\alpha)<\frac{C'}{M}$.
To do so, we will cover this event with smaller failure events $\mathcal{E}_k$ which correspond to members of a $\delta$-net.
First, define
\begin{equation*}
G_\delta(v)
:=\{\varphi\in\mathbb{C}^M:|\langle v,\varphi\rangle|\geq 3\delta\mbox{ and }\|\varphi\|\leq 2\}.
\end{equation*}
Then for each member $v$ of a given $\delta$-net $\mathcal{N}_\delta$ of the unit sphere in $\mathbb{C}^M$, define the failure event
\begin{equation*}
\mathcal{E}_v
:=\Big\{\mbox{less than $\alpha n$ columns of $\Phi$ lie in $G_\delta(v)$}\Big\}.
\end{equation*}
We claim that $\{\mathrm{PU}(\Phi;\alpha)<\delta^2\}\subseteq\bigcup_{v\in\mathcal{N}_\delta}\mathcal{E}_v$.
To see this, take an outcome $\omega\in(\bigcup_{v\in\mathcal{N}_\delta}\mathcal{E}_v)^\mathrm{c}$.  
By the definition of $\mathcal{N}_\delta$, we have that for every unit vector $x\in\mathbb{C}^M$, there exists $v\in\mathcal{N}_\delta$ such that $\|v-x\|\leq\delta$.
This gives a mapping $x\mapsto v_x$, and we let $\mathcal{I}_x$ denote the indices of columns of $\Phi$ which lie in $G_\delta(v_x)$.
Since $\omega\in\bigcap_{v\in\mathcal{N}_\delta}\mathcal{E}_v^\mathrm{c}$, we have $|\mathcal{I}_x|\geq\alpha n$ for every $x$.
Moreover, by the triangle inequality, Cauchy-Schwarz inequality, and the definition of $G_\delta$, every $i\in\mathcal{I}_x$ satisfies
\begin{equation*}
|\langle x,\varphi_i\rangle|
\geq |\langle v_x,\varphi_i\rangle|-|\langle v_x-x,\varphi_i\rangle|
\geq |\langle v_x,\varphi_i\rangle|-\|v_x-x\|\|\varphi_i\|
\geq 3\delta-2\|v_x-x\|
\geq \delta.
\end{equation*}
Applying this inequality then gives
\begin{equation*}
\mathrm{PU}(\Phi;\alpha)
=\min_{\substack{x\in\mathbb{C}^M\\\|x\|=1}}\max_{\substack{\mathcal{I}\subseteq\{1,\ldots,n\}\\|\mathcal{I}|\geq\alpha n}}\min_{i\in\mathcal{I}}|\langle x,\varphi_i\rangle|^2
\geq \min_{\substack{x\in\mathbb{C}^M\\\|x\|=1}}\min_{i\in\mathcal{I}_x}|\langle x,\varphi_i\rangle|^2
\geq\delta^2.
\end{equation*}
Thus $\omega\in\{\mathrm{PU}(\Phi;\alpha)<\delta^2\}^\mathrm{c}$, thereby proving our claim.
Continuing, the union bound gives
\begin{equation*}
\mathrm{Pr}\Big(\mathrm{PU}(\Phi;\alpha)<\delta^2\Big)
\leq\sum_{v\in\mathcal{N}_\delta}\mathrm{Pr}(\mathcal{E}_v)
=|\mathcal{N}_\delta|\cdot\mathrm{Pr}(\mathcal{E}_v),
\end{equation*}
where the equality follows from the fact that the columns of $\Phi$ are independent with rotationally symmetric probability distributions.
It therefore suffices to bound both $|\mathcal{N}_\delta|$ and $\mathrm{Pr}(\mathcal{E}_v)$.

To bound $|\mathcal{N}_\delta|$, we follow a standard argument, found in the proof of Lemma~5.2 in~\cite{Vershynin:11}.
Let $\mathcal{N}_\delta$ be a maximal $\delta$-packing of points on the unit sphere in $\mathbb{C}^M$.
Since the packing is maximal, it follows that $\mathcal{N}_\delta$ is a $\delta$-net.
To count these points, map them into $\mathbb{R}^{2M}$ according to $f\colon v\mapsto(\mathrm{Re}~v,\mathrm{Im}~v)$.
Consider the open balls of radius $\frac{\delta}{2}$ centered at each $f(v)$.
Note that the (disjoint) union of these balls is contained in the ball of radius $1+\frac{\delta}{2}$ centered at the origin.
Thus, a volume comparison gives
\begin{equation*}
|\mathcal{N}_\delta|\cdot C(\tfrac{\delta}{2})^{2M}
=\mathrm{Vol}\bigg(\bigsqcup_{v\in\mathcal{N}_\delta}B(f(v),\tfrac{\delta}{2})\bigg)
\leq\mathrm{Vol}\Big(B(0,1+\tfrac{\delta}{2})\Big)
=C(1+\tfrac{\delta}{2})^{2M},
\end{equation*}
thereby implying $|\mathcal{N}_\delta|\leq(\frac{2}{\delta}+1)^{2M}$.

To bound $\mathrm{Pr}(\mathcal{E}_v)$, note that
\begin{equation}
\label{eq.binomial prob}
\mathrm{Pr}(\mathcal{E}_v)
=\mathrm{Pr}\bigg(\sum_{i=1}^n 1_{\{\varphi_i\not\in G_\delta(v)\}}\geq(1-\alpha)n\bigg).
\end{equation}
As such, we first consider the success probability of these Bernoulli random variables:
\begin{equation}
\label{eq.bernoulli prob}
\mathrm{Pr}\Big(\varphi_i\not\in G_\delta(v)\Big)
\leq \mathrm{Pr}\Big(|\langle v,\varphi_i\rangle|<3\delta\Big)+\mathrm{Pr}\Big(\|\varphi_i\|>2\Big).
\end{equation}
By the rotational symmetry of the distribution of $\varphi_i$, we may take $v$ to be the first identity basis element $e_1$ without changing the probability.
Next, since $\varphi_i=\sum_{m=1}^M(a_m+\mathrm{i}b_m)e_m$ with the $a_m$'s and $b_m$'s independent with distribution $\mathcal{N}(0,\frac{1}{2M})$, we have
\begin{equation*}
\mathrm{Pr}\Big(|\langle v,\varphi_i\rangle|<3\delta\Big)
=\mathrm{Pr}\Big(\sqrt{a_1^2+b_1^2}<3\delta\Big)
\leq\mathrm{Pr}\Big(\min\{a_1^2,b_1^2\}<\tfrac{9}{2}\delta^2\Big)
\leq2\mathrm{Pr}\Big(a_1^2<\tfrac{9}{2}\delta^2\Big),
\end{equation*}
where the last step is by the union bound.
Since $a_1$'s density function is $\leq\sqrt{\frac{M}{\pi}}$, it follows that
\begin{equation}
\label{eq.bernoulli prob 1}
\mathrm{Pr}\Big(|\langle v,\varphi_i\rangle|<3\delta\Big)
\leq2\mathrm{Pr}\Big(|a_1|<\tfrac{3}{\sqrt{2}}\delta\Big)
\leq 12\delta\sqrt{\tfrac{M}{2\pi}}.
\end{equation}
For the other term in \eqref{eq.bernoulli prob}, note that $2M\|\varphi_i\|^2$ is a sum of independent standard Gaussian random variables.
Applying Lemma~1 of~\cite{LaurentM:00} then gives that for every $t>0$,
\begin{equation*}
\mathrm{Pr}\Big(2M\|\varphi_i\|^2\geq\sqrt{8Mt}+2t+2M\Big)
\leq\mathrm{e}^{-t}.
\end{equation*}
Thus, taking $t=\frac{M}{2}$ gives
\begin{equation}
\label{eq.bernoulli prob 2}
\mathrm{Pr}\Big(\|\varphi_i\|>2\Big)
=\mathrm{Pr}\Big(2M\|\varphi_i\|^2>8M\Big)
\leq\mathrm{Pr}\Big(2M\|\varphi_i\|^2\geq 5M\Big)
\leq\mathrm{e}^{-M/2}.
\end{equation}
Substituting \eqref{eq.bernoulli prob 1} and \eqref{eq.bernoulli prob 1} into \eqref{eq.bernoulli prob} then gives
\begin{equation}
\label{eq.bernoulli prob final}
\mathrm{Pr}\Big(\varphi_i\not\in G_\delta(v)\Big)
\leq 12\delta\sqrt{\tfrac{M}{2\pi}}+\mathrm{e}^{-M/2}.
\end{equation}
Now, to bound \eqref{eq.binomial prob}, we will apply Hoeffding's inequality~\cite{Hoeffding:63}, which says that the tail probability of a sum of independent Bernoulli random variables $X_i$, each with success probability $p$, has the following bound:
\begin{equation*}
\mathrm{Pr}\bigg(\sum_{i=1}^nX_i\geq n(p+t)\bigg)\leq\mathrm{e}^{-2nt}.
\end{equation*}
Also, note that replacing $p$ in the left-hand side above with some $p'\geq p$ will not increase the probability.
As such, taking $p'$ to be the right-hand side of \eqref{eq.bernoulli prob final} and $t=1-\alpha-p'$, we have
\begin{equation*}
\label{eq.binomial prob final}
\mathrm{Pr}(\mathcal{E}_v)
=\mathrm{Pr}\bigg(\sum_{i=1}^n 1_{\{\varphi_i\not\in G_\delta(v)\}}\geq(1-\alpha)n\bigg)
\leq\exp\bigg(-2n\Big((1-\alpha)-\Big(12\delta\sqrt{\tfrac{M}{2\pi}}+\mathrm{e}^{-M/2}\Big)\Big)\bigg).
\end{equation*}

Now that we have bounds on both $|\mathcal{N}_\delta|$ and $\mathrm{Pr}(\mathcal{E}_v)$, we note that we specifically wish to show that 
\begin{equation*}
\mathrm{Pr}\Big(\mathrm{PU}(\Phi;\alpha)<\tfrac{C'}{M}\Big)
<c_0\mathrm{e}^{-c_1M}
\end{equation*}
for some constants $c_0,c_1>0$.
Considering the above analysis and taking logarithms, it suffices to have
\begin{equation*}
\log|\mathcal{N}_\delta|+\log\mathrm{Pr}(\mathcal{E}_v)
\leq 2M\log(\tfrac{2}{\delta}+1)-2n(1-\alpha)+2n\Big(12\delta\sqrt{\tfrac{M}{2\pi}}+\mathrm{e}^{-M/2}\Big)
<\log c_0-c_1M
\end{equation*}
with $\delta^2=\frac{C'}{M}$.
Since $\alpha<1-\frac{1}{2C}$ by assumption, picking $C'<\frac{\pi}{72}(1-\frac{1}{2C}-\alpha)^2$ will make $2n(1-\alpha)$ the dominant term in the above inequality, thereby proving the result.
\qquad
\end{proof}

\subsection{Removing large vertices}

In this section, we prove how well we can remove the vertices with the largest noisy intensity measurements.
We start with a lemma:

\begin{lemma}
\label{lemma.large vertices}
Pick $C>0$ and $n\geq\mathrm{e}^{C/8}M\log M$.
Draw the entries of an $M\times n$ matrix $\Phi=\{\varphi_i\}_{i=1}^n$ independently from $\mathbb{C}\mathcal{N}(0,\frac{1}{M})$.
Taking $\beta:=3\mathrm{e}^{-C/8}$, then with overwhelming probability,
\begin{equation}
\label{eq.few vertices to delete}
\#\bigg\{i:|\langle x,\varphi_i\rangle|^2>\frac{C}{M}\bigg\}<\beta n
\end{equation}
for every unit norm $x\in\mathbb{C}^M$.
\end{lemma}

\begin{proof}
Take a $\delta$-net $\mathcal{N}_\delta$ of the unit sphere in $\mathbb{C}^M$.
Then for every unit vector $x$, there is a $v_x\in\mathcal{N}_\delta$ such that
\begin{equation*}
|\langle x,\varphi_i\rangle|
\leq|\langle v_x,\varphi_i\rangle|+|\langle x-v_x,\varphi_i\rangle|
\leq|\langle v_x,\varphi_i\rangle|+\|x-v_x\|
\leq|\langle v_x,\varphi_i\rangle|+\delta.
\end{equation*}
Taking $\delta=\frac{1}{2}\sqrt{\frac{C}{M}}$, we then have that $|\langle x,\varphi_i\rangle|^2>\frac{C}{M}$ implies
$|\langle v_x,\varphi_i\rangle|^2>\frac{C}{4M}$.
Recall that we wish to bound the probability of the event $\mathcal{E}$ that \eqref{eq.few vertices to delete} is violated for some unit vector $x$.
To this end, the above implication allows us to focus on a finite set of points:
\begin{equation*}
\mathrm{Pr}(\mathcal{E})
\leq\mathrm{Pr}\bigg(\exists v\in\mathcal{N}_\delta\mbox{ s.t. }\sum_{i=1}^n1_{\{|\langle v,\varphi_i\rangle|^2>\frac{C}{4M}\}}\geq \beta n\bigg).
\end{equation*}
As discussed in the proof of Lemma~\ref{th:RadialSpreadConstant_bigLemma}, we may take $|\mathcal{N}_\delta|\leq(\frac{2}{\delta}+1)^{2M}$, and so the union bound and the symmetric distribution of $\Phi$ both give
\begin{equation}
\label{eq.before hoeffding}
\mathrm{Pr}(\mathcal{E})
\leq \Big(4\sqrt{\tfrac{M}{C}}+1\Big)^{2M}\mathrm{Pr}\bigg(\sum_{i=1}^n1_{\{|\langle v,\varphi_i\rangle|^2>\frac{C}{4M}\}}\geq \beta n\bigg).
\end{equation}
Similar to the proof of Lemma~\ref{th:RadialSpreadConstant_bigLemma}, we will bound the probability on the right-hand side using Hoeffding's inequality.
First, we note that the symmetric distribution of $\varphi_i=\sum_{m=1}^M(a_m+\mathrm{i}b_m)e_m$ gives that
\begin{equation*}
\mathrm{Pr}\Big(|\langle v,\varphi_i\rangle|^2>\tfrac{C}{4M}\Big)
=\mathrm{Pr}\Big(|\langle e_1,\varphi_i\rangle|^2>\tfrac{C}{4M}\Big)
=\mathrm{Pr}\Big(a_1^2+b_1^2>\tfrac{C}{4M}\Big)
\leq\mathrm{Pr}\Big(a_1^2>\tfrac{C}{8M}\Big)+\mathrm{Pr}\Big(b_1^2>\tfrac{C}{8M}\Big),
\end{equation*}
where the last step is by the union bound.
We continue, using the fact that $a_1$ and $b_1$ are both distributed as $\mathcal{N}(0,\frac{1}{2M})$:
\begin{equation*}
\mathrm{Pr}\Big(|\langle v,\varphi_i\rangle|^2>\tfrac{C}{4M}\Big)
\leq2\mathrm{Pr}\Big(a_1^2>\tfrac{C}{8M}\Big)
=2\mathrm{Pr}\Big(z^2>\tfrac{C}{4}\Big)
\leq 2\mathrm{e}^{-C/8},
\end{equation*}
where $z$ is a standard Gaussian random variable.
With this, we now apply Hoeffding's inequality to \eqref{eq.before hoeffding}:
\begin{equation*}
\mathrm{Pr}(\mathcal{E})
\leq \Big(4\sqrt{\tfrac{M}{C}}+1\Big)^{2M}\mathrm{exp}\Big(-2n(\beta-2\mathrm{e}^{-C/8})\Big)
=\mathrm{exp}\bigg(2M\log\Big(4\sqrt{\tfrac{M}{C}}+1\Big)-2n\mathrm{e}^{-C/8}\bigg).
\end{equation*}
Since $n\geq\mathrm{e}^{C/8}M\log M$, then $2n\mathrm{e}^{-C/8}$ is the dominant term above, thereby proving the result.
\qquad
\end{proof}

\begin{theorem}
\label{theorem.large vertices}
Pick $C>0$ and $n\geq\mathrm{e}^{C/8}M\log M$.
Draw the entries of an $M\times n$ matrix $\Phi=\{\varphi_i\}_{i=1}^n$ independently from $\mathbb{C}\mathcal{N}(0,\frac{1}{M})$.
Given a vector $x\in\mathbb{C}^M$, consider
\begin{equation*}
z_i:=|\langle x,\varphi_i\rangle|^2+\nu_i,
\end{equation*}
where $\nu=\{\nu_i\}_{i=1}^n$ satisfies $\frac{\|\nu\|}{\|x\|^2}\leq\frac{C'}{\sqrt{M}}$.
Taking $\beta:=3\mathrm{e}^{-C/8}$, we have
\begin{equation}
\#\bigg\{i:z_i>\frac{2C}{M}\|x\|^2\bigg\}<2\beta n
\end{equation}
with overwhelming probability.
\end{theorem}

\begin{proof}
In counting large $z_i$'s, we identify which come from large or small inner products $|\langle x,\varphi_i\rangle|^2$.
First,
\begin{equation*}
\Big\{i:z_i>\tfrac{2C}{M}\|x\|^2\Big\}\cap\Big\{i:|\langle x,\varphi_i\rangle|^2>\tfrac{C}{M}\|x\|^2\Big\}
\end{equation*}
is of size $<\beta n$ with overwhelming probability by Lemma~\ref{lemma.large vertices}.
The rest of the large $z_i$'s have indices in
\begin{equation*}
\mathcal{K}:=\Big\{i:z_i>\tfrac{2C}{M}\|x\|^2\Big\}\cap\Big\{i:|\langle x,\varphi_i\rangle|^2\leq\tfrac{C}{M}\|x\|^2\Big\}.
\end{equation*}
To count these, note that $\nu_i=z_i-|\langle x,\varphi_i\rangle|^2\geq\frac{C}{M}\|x\|^2$, and so
\begin{equation*}
|\mathcal{K}|\tfrac{C^2}{M^2}\|x\|^4
\leq\sum_{i\in\mathcal{K}}|\nu_i|^2
\leq\|\nu\|^2
\leq\tfrac{(C')^2}{M}\|x\|^4.
\end{equation*}
Rearranging then reveals that $|\mathcal{K}|=\mathcal{O}(M)$, meaning $\mathcal{K}$ has fewer than $\beta n\geq 3M\log M$ members when $M$ is sufficiently large.
\qquad
\end{proof}

\subsection{Main result}

This section proves the main result of the paper, which we restate here:

\begin{theorem}
Pick $N\sim CM\log M$ with $C$ sufficiently large, and take $\{\varphi_\ell\}_{\ell=1}^N=\Phi_V\cup\Phi_E$ defined in the measurement design of Section~\ref{section.noisy}.
Then there exist constants $C',K>0$ such that the following guarantee holds for all $x\in\mathbb{C}^M$ with overwhelming probability:
Consider measurements of the form
\begin{equation*}
z_\ell:=|\langle x,\varphi_\ell\rangle|^2+\nu_\ell.
\end{equation*}
If the noise-to-signal ratio satisfies $\mathrm{NSR}:=\frac{\|\nu\|}{\|x\|^2}\leq\frac{C'}{\sqrt{M}}$, then the phase retrieval procedure of Section~\ref{section.noisy} produces an estimate $\tilde{x}$ from $\{z_\ell\}_{\ell=1}^N$ with squared relative error
\begin{equation*}
\frac{\|\tilde{x}-\mathrm{e}^{\mathrm{i}\theta}x\|^2}{\|x\|^2}
\leq K\sqrt{\frac{M}{\log M}}~\mathrm{NSR}
\end{equation*}
for some phase $\theta\in [0,2\pi)$.
\end{theorem}

\begin{proof}
We will prove the result by considering the steps of our phase retrieval process in reverse order.
In the last step, we have the following estimates of $\langle x,\varphi_i\rangle$ for every vertex $i\in V''\subseteq V$ which survives our graph-pruning and large-vertex-removing processes:
\begin{equation*}
y_i
:=\mathrm{e}^{\mathrm{i}\theta}\langle x,\varphi_i\rangle+\delta_i
=(\mathrm{e}^{\mathrm{i}\theta}\Phi_{V''}^*x+\delta)_i.
\end{equation*}
Here, $\theta$ is a global phase which is calculated in the proof of Theorem~\ref{thm:ang_synch_appendix}.
From these estimates, we reconstruct by finding the least-squares estimate of $x$: 
\begin{equation*}
\tilde{x}
:=(\Phi_{V''}\Phi_{V''}^*)^{-1}\Phi_{V''}y
=\mathrm{e}^{\mathrm{i}\theta}x+(\Phi_{V''}\Phi_{V''}^*)^{-1}\Phi_{V''}\delta.
\end{equation*}
As such, we have the following bound on the reconstruction error:
\begin{equation}
\label{eq.reconstruction error 1}
\|\tilde{x}-\mathrm{e}^{\mathrm{i}\theta}x\|
=\|(\Phi_{V''}\Phi_{V''}^*)^{-1}\Phi_{V''}\delta\|
\leq\|(\Phi_{V''}\Phi_{V''}^*)^{-1}\Phi_{V''}\|_2\|\delta\|
=\frac{\|\delta\|}{\sigma_\mathrm{min}(\Phi_{V''}^*)},
\end{equation}
where the last equality holds with probability $1$, specifically, in the event that $\sigma_\mathrm{min}(\Phi_{V''}^*)>0$.

To continue this bound, recall that the large-vertex-removing process ensured $|V''|=\kappa n$.
We claim there exists a constant $c_0>0$ such that with overwhelming probability, $\sigma_\mathrm{min}(\Phi_{V''}^*)\geq c_0\sqrt{\log M}$ for every $V''\subseteq V$ with $|V''|=\kappa n$.
To prove this, we leverage the complex version of Corollary~5.35 in~\cite{Vershynin:11}, which gives that for a fixed $V''\subseteq V$, there exist $c,c'>0$ such that 
\begin{equation*}
\mathrm{Pr}\Big(\sigma_\mathrm{min}(\Phi_{V''}^*)<\tfrac{1}{\sqrt{M}}\big(\sqrt{\kappa n}-\sqrt{M}-\sqrt{tn}\big)\Big)
\leq c\mathrm{e}^{-c'tn}.
\end{equation*}
Performing a union bound over the $\binom{n}{\kappa n}=\binom{n}{(1-\kappa)n}\leq(\frac{\mathrm{e}}{1-\kappa})^{(1-\kappa)n}$ choices for $V''$ then gives
\begin{equation*}
\mathrm{Pr}\Big(\exists V''\subseteq V,~|V''|=\kappa n\mbox{ s.t. }\sigma_\mathrm{min}(\Phi_{V''}^*)<\tfrac{1}{\sqrt{M}}\big(\sqrt{\kappa n}-\sqrt{M}-\sqrt{tn}\big)\Big)
\leq c~\!\mathrm{exp}\Big(c'n\big(-\tfrac{t}{2}+(1-\kappa)\log(\tfrac{\mathrm{e}}{1-\kappa})\big)\Big).
\end{equation*}
Provided $\frac{\kappa}{2}>(1-\kappa)\log(\frac{\mathrm{e}}{1-\kappa})$, which occurs whenever $\kappa\geq0.86$, then taking $t:=\frac{\kappa}{2}+(1-\kappa)\log(\frac{\mathrm{e}}{1-\kappa})$ will ensure both $t<\kappa$ and $\frac{t}{2}>(1-\kappa)\log(\frac{\mathrm{e}}{1-\kappa})$.
Therefore, this choice for $t$ coupled with the fact that $n=CM\log M$ proves the claim.
Continuing \eqref{eq.reconstruction error 1} then gives
\begin{equation}
\label{eq.reconstruction error 2}
\|\tilde{x}-\mathrm{e}^{\mathrm{i}\theta}x\|
\leq\frac{\|\delta\|}{c_0\sqrt{\log M}}.
\end{equation}

Next, we wish to bound $\|\delta\|$.
By definition, we have
\begin{equation*}
\delta_i
=y_i-\mathrm{e}^{\mathrm{i}\theta}\langle x,\varphi_i\rangle
=\sqrt{z_i}\mathrm{e}^{\mathrm{i}~\!\mathrm{arg}(u_i)}-\mathrm{e}^{\mathrm{i}\theta}\langle x,\varphi_i\rangle.
\end{equation*}
Note that the above square root operates under the assumption that $z_i\geq 0$ for each $i\in V''$, which is ensured when we prune for reliability.
Denote $\xi_i:=\sqrt{z_i}-|\langle x,\varphi_i\rangle|$.
Then by the triangle inequality, we have
\begin{equation*}
|\delta_i|
=\Big|\sqrt{z_i}\mathrm{e}^{\mathrm{i}~\!\mathrm{arg}(u_i)}-\sqrt{z_i}\mathrm{e}^{\mathrm{i}(\theta+\mathrm{arg}(\langle x,\varphi_i\rangle))}+\xi_i\mathrm{e}^{\mathrm{i}(\theta+\mathrm{arg}(\langle x,\varphi_i\rangle))}\Big|
\leq \sqrt{z_i}\Big|\mathrm{e}^{\mathrm{i}~\!\mathrm{arg}(u_i)}-\mathrm{e}^{\mathrm{i}(\theta+\mathrm{arg}(\langle x,\varphi_i\rangle))}\Big|+|\xi_i|.
\end{equation*}
Next, factoring $\mathrm{e}^{\mathrm{i}(\theta+\mathrm{arg}(\langle x,\varphi_i\rangle))}$ and applying \eqref{eq.norm equivalence t i} yields
\begin{equation*}
|\delta_i|
\leq \sqrt{z_i}\big\|\mathrm{arg}(u_i)-\mathrm{arg}(\langle x,\varphi_i\rangle)-\theta\big\|_\mathbb{T}+|\xi_i|.
\end{equation*}
For any $a,b\geq0$, then since $0\leq(a-b)^2=a^2-2ab+b^2$, we have $(a+b)^2=a^2+2ab+b^2\leq 2(a^2+b^2)$.
Applying this inequality to the right-hand side above then gives
\begin{equation*}
|\delta_i|^2
\leq 2z_i\big\|\mathrm{arg}(u_i)-\mathrm{arg}(\langle x,\varphi_i\rangle)-\theta\big\|_\mathbb{T}^2+2\xi_i^2.
\end{equation*}
Similarly, for any $a,b\geq0$, then since $ab\geq\min\{a^2,b^2\}$, we have $(a-b)^2=a^2-2ab+b^2\leq|a^2-b^2|$.
Applying this to $\xi_i^2=(\sqrt{z_i}-|\langle x,\varphi_i\rangle|)^2$ then gives $\xi_i^2\leq|\nu_i|$.
Also by Theorem~\ref{theorem.large vertices}, our large-vertex-removing process ensures that $z_i\leq\frac{2c_1}{M}\|x\|^2$ for every $i\in V''$.
Combined, these facts imply
\begin{equation}
\label{eq.reconstruction error 3}
\|\delta\|^2
\leq \frac{4c_1}{M}\|x\|^2\sum_{i\in V''}\big\|\mathrm{arg}(u_i)-\mathrm{arg}(\langle x,\varphi_i\rangle)-\theta\big\|_\mathbb{T}^2+2\|\nu_V\|_1.
\end{equation}
Applying Theorem~\ref{thm:ang_synch_appendix}, Definition~\ref{definition:NOISYPhaselessRadialSpreadConst} and Theorem~\ref{th:PUN_bigTheorem} further gives
\begin{equation}
\label{eq.reconstruction error 4}
\sum_{i\in V''}\big\|\mathrm{arg}(u_i)-\mathrm{arg}(\langle x,\varphi_i\rangle)-\theta\big\|_\mathbb{T}^2
\leq\frac{c_2\|\varepsilon\|^2}{\tau^2P^2}
\leq\frac{c_2\|\varepsilon\|^2}{\tau^2\|x\|^4\mathrm{PUN}(\Phi_V;\alpha,\frac{\varepsilon}{\|x\|^2})^2}
\leq\frac{c_2M^2\|\varepsilon\|^2}{c_3^2\tau^2\|x\|^4}.
\end{equation}
Recalling the definition of $\varepsilon$, we have
\begin{equation*}
\varepsilon_{ij}
=\bigg(\frac{1}{3}\sum_{k=0}^2\zeta^k\Big(|\langle x,\varphi_i+\zeta^k\varphi_j\rangle|^2+\nu_{ijk}\Big)\bigg)-\overline{\langle x,\varphi_i\rangle}\langle x,\varphi_j\rangle
=\frac{1}{3}\sum_{k=0}^2\zeta^k\nu_{ijk},
\end{equation*}
and so, letting $E'$ denote the edges which remained after pruning for connectivity, the Cauchy-Schwarz inequality gives
\begin{equation}
\label{eq.reconstruction error 5}
\|\varepsilon\|^2
=\sum_{\{i,j\}\in E'}|\varepsilon_{ij}|^2
=\sum_{\{i,j\}\in E'}\bigg|\frac{1}{3}\sum_{k=0}^2\zeta^k\nu_{ijk}\bigg|^2
\leq\frac{1}{9}\sum_{\{i,j\}\in E'}\bigg(\sum_{k=0}^2|\zeta^{-k}|^2\bigg)\bigg(\sum_{k=0}^2|\nu_{ijk}|^2\bigg)
\leq\frac{1}{3}\|\nu\|^2.
\end{equation}
Finally, we combine \eqref{eq.reconstruction error 2}, \eqref{eq.reconstruction error 3}, \eqref{eq.reconstruction error 4} and \eqref{eq.reconstruction error 5}, along with $\|\nu_V\|_1\leq\sqrt{n}\|\nu_V\|\leq\sqrt{2CM\log M}\|\nu\|$:
\begin{equation*}
\frac{\|\tilde{x}-\mathrm{e}^{\mathrm{i}\theta}x\|^2}{\|x\|^2}
\leq\frac{4c_1c_2}{3c_0^2c_3^2\tau^2}\frac{M}{\log M}\frac{\|\nu\|^2}{\|x\|^4}+\frac{2\sqrt{2C}}{c_0^2}\sqrt{\frac{M}{\log M}}\frac{\|\nu\|}{\|x\|^2}
\leq K\sqrt{\frac{M}{\log M}}\frac{\|\nu\|}{\|x\|^2},
\end{equation*}
for some constant $K$.
\qquad
\end{proof}

\subsection*{Acknowledgments}
The authors thank the anonymous referees for providing thoughtful suggestions that led to a more complete discussion of the context of our results.
The authors also thank Prof.\ Amit Singer for insightful discussions.
B.~Alexeev was supported by the NSF Graduate Research Fellowship under Grant No.\ DGE-0646086, A.S.~Bandeira was supported by NSF Grant No.~DMS-0914892, M.~Fickus was supported by NSF Grant No.~DMS-1042701 and AFOSR Grant Nos.~F1ATA01103J001 and F1ATA00183G003, and D.G.~Mixon was supported by the A.B.~Krongard Fellowship.
The views expressed in this article are those of the authors and do not reflect the official policy or position of the United States Air Force, Department of Defense, or the U.S.~Government.

\end{document}